\documentclass[letterpaper,twocolumn,10pt]{article}

\def\cameraReady{} 

\usepackage[compact]{titlesec}
\usepackage{usenix}
\usepackage{tikz}
\usepackage{amsmath}
\usepackage{cleveref}
\usepackage{booktabs}

\usepackage{listings}
\usepackage{amsmath}
\usepackage{comment}
\usepackage[skins, breakable]{tcolorbox}
\usepackage{multirow}
\usepackage{amsthm}
\usepackage{amsfonts}
\usepackage{graphicx}
\usepackage{appendix}
\usepackage{xspace}
\usepackage{color}
\usepackage{tikz}
\usepackage{pgfplots}
\usepackage{textcomp}
\usepackage{enumitem}
\usepackage{url}
\usepackage{subcaption}
\usepackage[T1]{fontenc}
\usepackage{changepage}
\usepackage{wrapfig}
\usepackage{tabularx}
\usepackage{hhline}
\usepackage{algorithm}
\usepackage[noend]{algpseudocode}

\newtheorem{lemma}{Lemma}
\newtheorem{theorem}{Theorem}
\newtheorem{definition}{Definition}

\newtheorem{observation}{Observation}

\newcommand{\threechain}{$3$-chain DiemBFT\xspace}
\newcommand{\vaba}{$2$-chain VABA\xspace}

\newcommand{\async}{\texttt{Ditto}\xspace}
\newcommand{\sync}{\texttt{Jolteon}\xspace}

\newcommand{\latency}{block-commit latency\xspace}
\newcommand{\etelatency}{end-to-end latency\xspace}

\newcommand{\timeout}{\texttt{timeout}\xspace}

\newcommand{\hash} {H}
\newcommand{\hqc} {qc_{high}}
\newcommand{\rank} {rank}

\newcommand{\mode} {\texttt{fallback}}
\newcommand{\mvba} {\texttt{MVBA}\xspace}
\newcommand\sig[1]{\langle #1 \rangle\xspace}
\newcommand\thsig[1]{\{ #1 \}\xspace}
\newcommand\block[1]{[ #1 ]\xspace}
\newcommand{\timeoutvalue}{$\tau$\xspace}

\algnewcommand{\LineComment}[1]{\Statex \hskip\ALG@thistlm {\color{gray}\textrm{// #1}}}

\newtcolorbox{mybox}[1][]{
enhanced,
colback=white,
boxsep=0pt,
#1
} 

\usepackage[textsize=tiny]{todonotes}

\renewcommand{\paragraph}[1]{\smallskip\noindent{\textbf{#1}}}

\newcommand\blue[1]{\textcolor{blue}{#1}\xspace}

\def\first{({i})\xspace}
\def\second{({ii})\xspace}
\def\third{({iii})\xspace}
\def\fourth{({iv})\xspace}

\microtypecontext{spacing=nonfrench}

\begin{document}

\date{}

\title{\Large 
Jolteon and Ditto: Network-Adaptive Efficient Consensus with Asynchronous Fallback
}

\ifdefined\cameraReady
\author{
{\rm Rati Gelashvili}\\
Novi Research
\and
{\rm Lefteris Kokoris-Kogias}\\
Novi Research \& IST Austria
\and
{\rm Alberto Sonnino}\\
Novi Research
\and
{\rm Alexander Spiegelman}\\
Novi Research
\and
{\rm Zhuolun Xiang}\thanks{Lead author, part of work was done at Novi Research}
\\
University of Illinois at Urbana-Champaign
}
\else
\author{}
\fi

\maketitle

\begin{abstract}
Existing committee-based Byzantine state machine replication (SMR) protocols, typically deployed in production blockchains, face a clear trade-off:
(1) they either achieve linear communication cost in the happy path, but sacrifice liveness during periods of asynchrony, or (2) they are robust (progress with probability one) but pay quadratic communication cost.
We believe this trade-off is unwarranted since existing linear protocols still have asymptotic quadratic cost in the worst case.

We design \async, a Byzantine SMR protocol that enjoys the best of both worlds: optimal communication on and off the happy path (linear and quadratic, respectively) and progress guarantee under asynchrony and DDoS attacks. 
We achieve this by replacing the view-synchronization of partially synchronous protocols with an asynchronous fallback mechanism at no extra asymptotic cost.
Specifically, we start from HotStuff, a state-of-the-art linear protocol, and gradually build \async. As a separate contribution and an intermediate step, we design a 2-chain version of HotStuff, \sync, which leverages a quadratic view-change mechanism to reduce the latency of the standard 3-chain HotStuff.

We implement and experimentally evaluate all our systems. Notably, \sync's commit latency outperforms HotStuff by 200-300ms with varying system size. Additionally, \async adapts to the network and provides better performance than \sync under faulty conditions and better performance than VABA (a state-of-the-art asynchronous protocol) under faultless conditions. This proves our case that breaking the robustness-efficiency trade-off is in the realm of practicality.  
\end{abstract}
\section{Introduction}

\begin{table*}[h]
\centering
\resizebox{\linewidth}{!}{
\begin{tabular}{lcccc}
\toprule
& {\bf Problem} & {\bf Comm. Complexity} & {\bf Rounds} & {\bf Liveness} \\
\midrule

HotStuff~\cite{yin2019hotstuff}/DiemBFT~\cite{baudet2019state} & partially sync SMR & sync $O(n)$ & $7$ & not live if async \\
VABA~\cite{abraham2019asymptotically} & async BA & $O(n^2)$ & $E(16.5)$ & always live \\

\midrule

\sync & partially sync SMR & sync $O(n)$ & $5$ & not live if async \\
\async & async SMR with fast sync path & sync $O(n)$, async $O(n^2)$ & sync $5$, async $E(R+4)$ & always live \\

\bottomrule
\end{tabular}
}

\caption{
Theoretical comparison of our protocol implementations. For HotStuff/DiemBFT and our protocols, sync $O(n)$ assumes honest leaders. Communication complexity measures the cost per consensus decision (committed block). Rounds measure the \latency. $E(r)$ means $r$ rounds in expectation.
The async $E(R+4)$ latency of \async assumes a MVBA protocol of expected latency $R$. 
With the MVBA of~\cite{cryptoeprint:2024:1108} that has $9.5$ rounds, the expected async latency of \async is $E(13.5)$.
}

\label{table:comparison}
\end{table*}
The popularity of blockchain protocols has generated a surge in researching how to increase the efficiency and robustness of underlying consensus protocols used for agreement.
On the efficiency front, the focus has been on decreasing the communication complexity in the steady state
(``happy path''),
first to quasilinear~\cite{kogias2016enhancing} and ultimately to linear~\cite{yin2019hotstuff,gueta2019sbft}.
These protocols work in the eventually synchronous model and require a leader to aggregate proofs. However, handling leader failures or unexpected network delays requires quadratic communication and
if the network is asynchronous, there is no liveness guarantee.
On the robustness side, recent protocols~\cite{abraham2019asymptotically,lu2020dumbo,spiegelman2019ace} make progress by having each replica act as the leader and decide on a leader retroactively.
This inherently requires quadratic communication even under good network conditions when the adversary is strongly adaptive~\cite{abraham2019communication}.


We believe that in practice, we need the best of both worlds. An efficient happy path is beneficial for any production system, but for blockchains to support important (e.g. financial) infrastructure, robustness against asynchrony is also key.
First, unpredictable network delays are a common condition when running in a large-scale network environment, e.g. over the Internet. Second, the possibility of targeted DDoS attacks on the leaders of leader-based protocols motivates the leaderless nature of the asynchronous solutions.
Thus, we are interested if there are efficient systems that have a linear happy path and are robust against asynchrony.

This is an important question, posed as early as~\cite{kursawe2005optimistic},
and studied from a theoretical prospective~\cite{ramasamy2005parsimonious, spiegelman2020search},
but the existing blockchain systems still forfeit robustness for efficiency~\cite{yin2019hotstuff,gueta2019sbft,kogias2016enhancing}.
In this paper, we answer the above question with the first practical system, tailor-made to directly apply to the prominent HotStuff/DiemBFT~\cite{yin2019hotstuff,baudet2019state} family of protocols.
Our protocol, \async, combines the optimistic (good network conditions) efficient happy path with pessimistic (worst-case network conditions) liveness guarantees with no extra asymptotic communication cost.
\async is based on the key observation that when there is asynchrony or failures, the protocols with linear happy path still pay the quadratic cost, same as state-of-the-art asynchronous protocols (e.g., VABA~\cite{abraham2019asymptotically} or Dumbo~\cite{lu2020dumbo}) that provide significantly more robustness.
Specifically, \async replaces the pacemaker of HotStuff/DiemBFT (a quadratic module that deals with view synchronization) with an asynchronous fallback.

In other words, instead of synchronizing views that will anyway fail (timeout) due to asynchrony or faults, we fall back to an asynchronous protocol that robustly guarantees progress at the cost of a single view-change. Furthermore, \async switches between the happy path and the fallback without overhead (e.g. additional rounds), and continues operating in the same pipelined fashion as HotStuff/DiemBFT.
 
Leveraging our observation further, we abandon the linear view-change~\cite{yin2019hotstuff} of HotStuff/DiemBFT, since the protocol anyway has a quadratic pacemaker~\cite{baudet2019state}.
We present 
\sync, a protocol that's a hybrid of HotStuff~\cite{yin2019hotstuff} and classical PBFT~\cite{castro1999practical}.
In particular, \sync preserves the structure of HotStuff and its linearity under good network conditions while reducing the steady state \latency by 30\% using a 2-chain commit rule. This decrease in latency comes at the cost of a quadratic view-change ($O(n)$ messages of $O(n)$ size).
As the pacemaker is already quadratic, as expected, this does not affect performance in our experiments.
We shared our findings with the Diem Team who is currently integrating \sync into the next release of DiemBFT.

Since \sync outperforms the original HotStuff in every scenario, we use it as the basis for \async.
As shown experimentally under optimistic network conditions both \sync and \async outperform HotStuff \latency.
Importantly, \async's performance is almost identical to \sync in this case whereas during attack \async performs almost identically to our implementation of VABA~\cite{abraham2019asymptotically}. 
Finally, the throughput of \async is 50\% better than VABA in the optimistic path and much better than HotStuff and \sync under faulty (dead) leaders (30-50\% better) or network instability (they drop to 0).

In Table~\ref{table:comparison} we show a theoretical comparison of our work to HotStuff and VABA, for which we implement and evaluate the performance experimentally in Section~\ref{sec:evaluation}.
More comprehensive comparisons and related work are given in Section~\ref{sec:relatedwork}.

\section{Preliminaries}
\label{sec:prelim}
We consider a permissioned system that consists of an adversary and $n$ replicas numbered $1, 2, \ldots, n$, where each replica has a public key certified by a public-key infrastructure (PKI).
The replicas have all-to-all reliable and authenticated communication channels controlled by the adversary.
We say a replica multicasts a message if it sends the message to all replicas.
We consider a dynamic adversary that can adaptively corrupt up to $f$ replicas, referred as \emph{Byzantine}. 
The rest of the replicas are called \emph{honest}.
The adversary controls the message delivery times, but we assume messages among honest replicas are eventually delivered. 

An execution of a protocol is 
{\em synchronous} if all message delays between honest replicas are bounded by $\Delta$;
is {\em asynchronous} if they are unbounded;
and is {\em partially synchronous} if there is a global stabilization time (GST) after which they are bounded by $\Delta$~\cite{dwork1988consensus}.
Without loss of generality, we let $n=3f+1$ where $f$ denotes the assumed upper bound on the number of Byzantine faults, which is the optimal worst-case resilience bound for asynchrony, partial synchrony~\cite{dwork1988consensus}, or asynchronous protocols with fast synchronous path~\cite{blum2020network}.

\paragraph{Cryptographic primitives and assumptions.}
We assume standard digital signature and public-key infrastructure (PKI), and use $\sig{m}_i$ to denote a message $m$ signed by replica $i$.
We also assume a threshold signature scheme, where a set of signature shares for message $m$ from $t$ (the threshold) distinct replicas can be combined into one threshold signature of the same length for $m$.
We use $\thsig{m}_i$ to denote a threshold signature share of a message $m$ signed by replica $i$.
We also assume a collision-resistant cryptographic hash function $\hash(\cdot)$ that can map an input of arbitrary size to an output of fixed size. 
For simplicity of presentation, we assume the above cryptographic schemes are ideal and a trusted dealer equips replicas with these cryptographic materials. The dealer assumption can be lifted using the protocol of~\cite{kokoris2020asynchronous}.

Any deterministic agreement protocol cannot tolerate even a single fault under asynchrony due to FLP~\cite{fischer1985impossibility}.
Our asynchronous fallback protocol generates distributed randomness following the approach of~\cite{cachin2005random}: the generated randomness of any given view is the hash of the unique threshold signature on the view number.

\paragraph{BFT SMR.}
A Byzantine fault-tolerant state machine replication protocol~\cite{castro1999practical} commits client transactions as a log akin to a single non-faulty server, and provides the following two guarantees:
\begin{itemize}[leftmargin=*,noitemsep,topsep=0pt]
    \item Safety. Honest replicas do not commit different transactions at the same log position. 
    \item Liveness. Each client transaction is eventually committed by all honest replicas.
\end{itemize}


\smallskip
Besides these two requirements, a validated BFT SMR protocol must also satisfy {\em external validity}~\cite{cachin2001secure}, requiring all committed transactions to be {\em externally valid}, i.e. satisfying some application-dependent predicate.
This can be accomplished by adding validity checks on the transactions before the replicas propose or vote, and for brevity, we omit the details and focus on the BFT SMR formulation defined above.
We assume that each client transaction will be repeatedly proposed by honest replicas until it is committed\footnote{For example, clients can send their transactions to all replicas, and the leader can propose transactions that are not yet included in the blockchain, in the order that they are submitted. With rotating leaders of HotStuff/DiemBFT and random leader election of the asynchronous fallback, the assumption can be guaranteed.}.
For most of the paper, we omit the client from the discussion and focus on replicas.

SMR protocols usually implement many instances of single-shot Byzantine agreement, but there are various approaches for ordering.
We focus on the chaining approach, used in HotStuff~\cite{yin2019hotstuff} and DiemBFT~\cite{baudet2019state}, in which each proposal references the previous one and each commit commits the entire prefix of the chain.

\begin{figure*}[t!]
    \centering
    \input{sections/diembft-code}
    \caption{DiemBFT in our terminology.}
    \label{fig:diembft}
\end{figure*}

\subsection{Terminology}
We present the terminology used throughout the paper.
For brevity, we omit explicit format checks in the protocol description. The honest replica will discard messages that are not correctly formatted according to the protocol specification. 
\begin{itemize}[leftmargin=*,noitemsep,topsep=0pt]
    \item {\em Round Number and View Number.} 
      The protocol proceeds in rounds and views and each replica keeps track of the current round number $r_{cur}$ and view number $v_{cur}$, which are initially set to $0$.
      Each view can have several rounds and it is incremented by $1$ after each asynchronous fallback.
      Each round $r$ has a designated leader $L_r$ that proposes a new block (defined below) of transactions in round $r$.
    Since DiemBFT and \sync  (\Cref{sec:jolteon}) do not have an asynchronous fallback, the view number is always $0$, but it is convenient to define it here so that our \async protocol (\Cref{sec:ditto}) uses the same terminology. 
    
    \item {\em Block Format.}
      A block is formatted as $B=\block{id, qc, r, v, txn}$ 
      where $qc$ is the quorum certificate (defined below) of $B$'s parent block in the chain, $r$ is the round number of $B$, $v$ is the view number of $B$, $txn$ is a batch of new transactions, and $id=\hash(qc,r,v,txn)$ is the unique hash digest of $qc,r,v,txn$.
      Note that when describing the protocol, it suffices to specify $qc$, $r$, and $v$ for a new block,
        since $txn,id$ follow the definitions.
    We will use $B.x$ to denote the element $x$ of $B$.
    
    \item {\em Quorum Certificate.}
    A \emph{quorum certificate (QC)} of some block $B$ is a threshold signature of a message that includes $B.id, B.r, B.v$, produced by combining the signature shares $\thsig{B.id, B.r, B.v}$ from a quorum of $n-f=2f+1$ replicas. 
    We say a block is \emph{certified} if there exists a QC for the block.
    Blocks are chained by QCs to form a blockchain, or block-tree if there are forks
    The round and view numbers of $QC$ for block $B$ are denoted by $QC.r$ and $QC.v$, which equals $B.r$ and $B.v$, respectively.
    A QC or a block of view number $v$ and round number $r$ has rank $\rank=(v, r)$, and QCs or blocks are compared lexicographically by their rank (i.e. first by the view number, then by the round number).
    The function $\max(\rank_1,\rank_2)$ returns the higher rank between $\rank_1$ and $\rank_2$, and
    $\max(qc_1,qc_2)$ returns the higher (ranked) QC between $qc_1$ and $qc_2$.
    In DiemBFT and \sync  (\Cref{sec:jolteon}) the view number is always 0, so QCs or blocks are compared by their round numbers; while in \async (\cref{sec:ditto}) they are compared by their rank.
    We use $\hqc$ to denote the highest quorum certificate.
    
    \item {\em Timeout Certificate.}
    A timeout message of round $r$ by a replica contains a signed $(r, \hqc)$ message, and a $TC$ of round $r-1$ if $\hqc.round\neq r-1$ (otherwise $TC=\bot$), i.e., $[\sig{r,\hqc}_i, TC]$.
    A timeout certificate (TC) is formed by a quorum of $n-f=2f+1$ timeout messages, containing $2f+1$ signed (round, highest QC) messages from these timeout messages.
    A valid TC of round $r$ should only contain signed $(r, *)$, and only contain highest QCs with round numbers $<r$, and this will be checked implicitly when a replica receives a TC.
\end{itemize}

\paragraph{Performance metrics.}
We consider message complexity and communication complexity per round and consensus decision. 
For the theoretical analysis of the latency, we consider \latency, i.e., the number of rounds for all honest replicas to commit a block since it is proposed.
For the empirical analysis, we measure the \etelatency, i.e., the time to commit a transaction since it is sent by a client.

\paragraph{Description of DiemBFT.}
The DiemBFT protocol (also known as LibraBFT)~\cite{baudet2019state} 
is a production version of HotStuff~\cite{yin2019hotstuff} with a synchronizer implementation (Pacemaker).
For better readability, Figure~\ref{fig:diembft} presents the DiemBFT protocol in our terminology, which we will refer to as HotStuff or DiemBFT interchangeably throughout the paper
There are two components of DiemBFT, a {\em Steady State protocol} that makes progress when the round leader is honest, and a {\em Pacemaker protocol} that advances round numbers either due to the lack of progress
or due to the current round being completed.
The leader $L_r$, upon entering round $r$, proposes a block $B$ that extends a block certified by the highest QC it knows about, $\hqc$. 
When receiving the first valid round-$r$ block from $L_r$, any replica tries to advance its current round number, update its highest locked
round and its highest QC, and checks if any block can be committed.
A block can be committed if it is the first block among $3$ adjacent certified blocks with consecutive round numbers. 
After the above steps, the replica votes for $B$ by sending a threshold signature share to the next leader $L_{r+1}$, if the voting rules are satisfied.
Then, when the next leader $L_{r+1}$ receives $2f+1$ such votes, it forms a QC of round $r$, enters round $r+1$, proposes the block for that round, and the above process is repeated.
When the timer of some round $r$ expires, the replica stops voting for that round and multicast a timeout message containing a threshold signature share for $r$ and its highest QC.
When any replica receives $2f+1$ such timeout messages, it forms a TC of round $r$, enters round $r+1$ and sends the TC to the (next) leader $L_{r+1}$.
When any replica receives a timeout or a TC, it tries to advance its current round number given the high-QCs (in the timeout or TC) or the TC, updates its highest locked round and its highest QC given the high-QCs, and checks if any block can be committed.
For space limitation, we omit the correctness proof of the DiemBFT protocol, which can be found in~\cite{baudet2019state}.


\section{\sync Design} \label{sec:jolteon}

\begin{figure*}[t!]
    \centering
    \input{sections/jolteon-code}
    \caption{\sync.}
    \label{fig:jolteon}
\end{figure*}

\begin{table*}[h]
\centering
\begin{tabular}{lcccc}
\toprule

& {\bf Latency} & {\bf steady state communication} &  {\bf view-change communication} &  {\bf view synchronization} \\
\midrule

DiemBFT & 7 messages & linear & linear &  quadratic \\
\sync & 5 messages & linear & quadratic & quadratic \\

\bottomrule
\end{tabular}
\caption{Theoretical comparison between DiemBFT and \sync.}
\label{table:DvJ}
\end{table*}

In this section, we describe how we turn DiemBFT into \sync\ -- a 2-chain version of DiemBFT.
The pseudocode is given in Figure~\ref{fig:jolteon}.
As mentioned previously, the quadratic cost of view-synchronization in leader-based consensus protocols, due to faulty leaders or asynchronous periods, is inherent.
While the linearity of HotStuff's view-change is a theoretical milestone, its practical importance is limited by this anyway quadratic cost of synchronization after bad views.

With this insight in mind, \sync uses a quadratic view-change protocol that allows a linear 2-chain commit rule in the steady state.
The idea is inspired by PBFT~\cite{castro1999practical} with each leader proving the safety of its proposal.
In the steady state each block extends the block from the previous round and providing the QC of the parent is enough to prove safety, hence the steady state protocol remains linear.
However, after a bad
round caused by asynchrony or a bad leader, proving the safety of extending an older QC requires the leader to prove that nothing more recent than the block of that QC is committed. 
To prove this, the leader uses the TC formed for view-changing the bad round. 
Recall that a TC for round $r$ contains $2f+1$ validators' $qc_{high}$ sent in timeout messages for round $r$. 
The leader attaches the TC to its proposal in round $r+1$ and extends the highest QC among the QCs in the TC (see the \textbf{Propose} rule in Figure~\ref{fig:jolteon}).

When a validator gets a proposal $B$, it first tries to advance its round number, then updates its $qc_{high}$ with $B.qc$ and checks the 2-chain commit rule for a possible commit.
Then, before voting,
it verifies that at least one of the following two conditions is satisfied:
\begin{itemize}[leftmargin=*,noitemsep,topsep=0pt]
    \item $B.r = B.qc.r+1$ or;
    \item $B.r=B.tc.r+1$ and\\ $B.qc.r\geq \max\{qc_{high}.r~|~ qc_{high}\in B.tc\}$
\end{itemize}
In other words, either $B$ contains the QC for the block of the previous round; or it contains at least the highest QC among the $2f+1$ QCs in the attached TC, which was formed to view-change the previous round.

\paragraph{Safety intuition.}
If the first condition is satisfied then $B$ directly extends the block from the previous round. 
Since at most one QC can be formed in a round, this means that no forks are possible, and voting for $B$ is safe.

The second condition is more subtle.
Note that by the 2-chain commit rule, if a block $B'$ is committed, then there exists a certified block $B''$ s.t. $B'. round + 1 = B''.round$.
That is, at least $f+1$ honest replicas vote to form the QC for $B''$ and thus set their $qc_{high}$ to be the QC for block $B'$ ($qc_{B'}$).
By quorum intersection and since replicas never decrease their $qc_{high}$, any future (higher round) TC contains a $qc_{high}$ that is at least as high as $qc_{B'}$. 
The second condition then guarantees that honest replicas only vote for proposals that extend the committed block $B'$.
Full proof can be found in Appendix~\ref{sec:jolteon-proof}.

\subsection{DiemBFT vs \sync.}

\paragraph{Efficiency.}
Table~\ref{table:DvJ} compares the efficiency of DiemBFT and \sync from a theoretical point of view.
Both protocols have linear communication complexity per round and per decision {\em under synchrony and honest leaders}, due to the leader-to-all communication pattern and the threshold signature scheme\footnote{The implementation of DiemBFT does not use threshold signatures, but for the theoretical comparison here we consider a version of DiemBFT that does.}.
The complexity of the Pacemaker to synchronize views (view synchronization in Table~\ref{table:DvJ}), for both protocols, under asynchrony or failures is quadratic due to the all-to-all timeout messages.
The complexity of proposing a block after a bad round that requires synchronization (view-change communication in Table~\ref{table:DvJ}) is linear for DiemBFT and quadratic for \sync.
This is because in DiemBFT such a proposal only includes $qc_{high}$, whereas in \sync it includes a TC containing $2f+1$ $qc_{high}$.
The \latency \emph{under synchrony and honest leaders} is $7 \Delta$ and $5 \Delta$ for DiemBFT and \sync, respectively, due to the $3$-chain ($2$-chain for \sync) commit rules. Each round in the commit chain requires two rounds trip times (upper bounded by $\Delta$), plus the new leader multicast the last QC of the chain that allows all honest replicas to learn about the chain and commit the block.

\paragraph{Limitations.}
During periods of asynchrony, or when facing DDoS attacks on the leaders, both protocols have {\em no liveness guarantees} -- the leaders' blocks cannot be received on time.
As a result, replicas keep multicasting timeout messages and advancing round numbers without certifying or committing blocks.
This is unavoidable~\cite{spiegelman2020search}: communication complexity of any deterministic partially synchronous Byzantine agreement protocol is unbounded before GST, even in failure-free executions.

Fortunately, in the next section, we show that it is possible to boost the liveness guarantee of DiemBFT, \sync,
by replacing the view-synchronization mechanism (pacemaker) with a fallback protocol that guarantees progress even under asynchrony.
Furthermore, the asynchronous fallback can be efficient. 
The protocol we propose in the next section has quadratic communication cost for fallback, which is the cost DiemBFT and \sync
pay to synchronize views anyway.



\section{\async Design} \label{sec:ditto}
\label{sec:fallback}

To strengthen the liveness guarantees of existing partially synchronous BFT protocols such as DiemBFT~\cite{baudet2019state} and \sync, we propose the protocol \async. 
\async has linear communication cost for the synchronous path, quadratic cost for the asynchronous path, and preserves liveness robustly in asynchronous network conditions.
The \async protocol is presented in Figure~\ref{fig:ditto}.

Our \async protocol uses Multi-valued Validated Byzantine Agreement (MVBA) in a black-box way.

\paragraph{Multi-valued Validated Byzantine Agreement}
\label{sec:pre:mvba}
Multi-valued validated Byzantine agreement (MVBA)~\cite{cachin2001secure} is a Byzantine fault-tolerant agreement protocol where a set of protocol replicas each with an input value can agree on the same value satisfying a predefined external predicate  $f(v):\{0,1\}^{|v|}\rightarrow \{0, 1\}$ globally known to all the replicas.
An MVBA protocol with predicate $f(\cdot)$ should provide the following guarantees except for negligible probability. 
\begin{itemize}[itemsep=0pt,topsep=4pt,leftmargin=*]
    \item \textbf{Agreement}: All honest replicas output the same value.
    \item \textbf{External Validity}: If an honest replica outputs $v$, then $v$ must be externally valid, i.e., $f(v)=1$.
    \item \textbf{Termination}: If all honest replicas input an externally valid value, all honest replicas eventually output.
\end{itemize}

Our protocol uses MVBA with the state-aware predicate~\cite{yurek2023long}, where the predicate $f(v, e)$ can have an additional state-dependent variable $e$ chosen by the replica when invoking the MVBA. With the state-aware predicate, a value $v$ is externally valid to an honest replica, if and only if $f(v, e)=1$ where $e$ is the state-dependent variable chosen by the honest replica when invoking the MVBA. The guarantees of MVBA with the state-aware predicate become the following.

\begin{itemize}[itemsep=0pt,topsep=4pt,leftmargin=*]
    \item \textbf{Agreement}: All honest replicas output the same value.
    \item \textbf{External Validity}: If an honest replica outputs $v$, then $v$ must be externally valid to at least one honest replica.
    \item \textbf{Termination}: If all honest replicas input a value that is externally valid to all honest replicas, all honest replicas eventually output.
\end{itemize}

Our protocol \async can directly use existing MVBA protocols in a black-box manner, by plugging in the state-aware predicate as defined in Figure~\ref{fig:ditto}. 
We refer the reader to~\cite{yurek2023long} for the argument as to why the agreement, termination, and external validity properties of MVBA still hold.

\paragraph{Protocol intuition.}
Our solution consists of a steady-state protocol, which is similar to that of \sync, and an asynchronous fallback protocol, which replaces the view-change of \sync.
The idea behind our fallback protocol is that, after entering the fallback, all replicas will invoke MVBA to commit a new block extending the committed blocks. Then, replicas will try the steady state again from the block committed by MVBA. 

Since this protocol has two paths, a synchronous fast path and an asynchronous fallback path, it is critical to ensure safety and liveness when the protocol transfers from one path to another. 
On a high level, our protocol ensures safety by always following the $1$-chain lock and $2$-chain commit rule from \sync, and by leveraging the agreement guarantee of MVBA. 
As for liveness, our protocol guarantees that either the sync path (same as \sync) makes progress, or enough replicas timeout the synchronous path and enter the asynchronous fallback.

\begin{figure*}[t]
    \centering
    
    \begin{mybox}
    In addition to the variables in Figure~\ref{fig:jolteon},
    each replica also keeps a boolean value $\mode$, initialized as $false$, to specify whether the replica is in a fallback.
    \vspace{2mm}

    \textbf{Steady State Protocol for Replica $i$}

    \begin{itemize}[leftmargin=*,noitemsep,topsep=0pt]
        \item\label{smr:propose} {\bf Propose.} 
        Upon entering round $r$, the leader $L_r$ multicasts a block {$B=\block{id, \hqc, r, v_{cur}, txn}$}.
        
        \item\label{smr:vote} {\bf Vote.} 
        Upon receiving the first proposal {$B=\block{id, qc, r,v, txn}$} from $L_{r}$.
        If $r=r_{cur}$, $v=v_{cur}$, $r>r_{vote}$ and {$r=qc.r+1$},
        vote for $B$ by sending the threshold signature share $\thsig{ id, r, v }_i$ to $L_{r+1}$, and update $r_{vote}\gets r$. 

        \item {\bf New QC.} Upon receiving a $qc$, replica $i$ performs the following steps.
        \begin{itemize}[leftmargin=*,noitemsep,topsep=0pt]
            \item {\bf Advance Round.} Replica $i$ updates its current round $r_{cur} \gets \max(r_{cur}, qc.r+1)$. 
            \item {\bf Lock.} ($1$-chain lock rule) Replica $i$ updates $\hqc\gets \max(\hqc,qc)$.
            \item {\bf Commit.} ($2$-chain commit rule) If there exists two adjacent certified blocks $B,B'$ {with the same view number}, replica $i$ commits $B$ and all its ancestors.
        \end{itemize}
        
    

        \item\label{async:timeout} {\bf Timeout.}
        Upon entering a new round, replica $i$ resets its timer to \timeoutvalue.
        
        When the timer expires, replica $i$ stops the timer as well as proposing or voting for any block of view $\leq v_{cur}$, and multicasts a \timeout message $\sig{ \thsig{ v_{cur} }_i, qc_{high} }_i$ where  $\thsig{ v_{cur} }_i$ is a threshold signature share.
    \end{itemize}
    \vspace{2mm}

    \textbf{Fallback Protocol for Replica $i$}

    Replica $i$ executes the following in event-driven manner.
    \begin{itemize}[leftmargin=*,noitemsep,topsep=0pt]

        \item\label{async:enter} {\bf Enter Fallback.}
        Upon receiving $2f+1$ \timeout messages of same view $v\geq v_{cur}$,
        replica $i$ updates
        $\mode\gets true$, 
        $v_{cur}\gets v+1$,
        and multicasts a message $(\texttt{proof}, v_{cur}, \hqc)$.

        \item\label{async:ack} {\bf Ack.} When $\mode=true$,
        upon receiving the first message $(\texttt{proof}, v, qc)$ from replica $j$ where $v=v_{cur}$,
        \begin{itemize}[noitemsep,topsep=0pt]
            \item if $qc.rank \geq \hqc.rank$, replica $i$ sends $\{\texttt{proof}, v_{cur}, qc\}_i$ back to replica $j$;

            \item if $qc.rank < \hqc.rank$, replica $i$ sends $(\texttt{higherQC}, v_{cur}, \hqc)$ back to replica $j$;
        \end{itemize}
        
        \item\label{async:enter} {\bf MVBA.} When $\mode=true$, and replica has not invoked $\mvba_{v_{cur}}$ yet,
        \begin{itemize}[noitemsep,topsep=0pt]
            \item upon aggregating a threshold signature  $\sigma$ on message $(\texttt{proof}, v_{cur}, \hqc)$, or
            \item upon receiving $(\texttt{higherQC}, v_{cur}, qc)$ where $qc.rank>\hqc.rank$, replica $i$ updates $\hqc\gets qc$ and sets $\sigma\gets\bot$,
        \end{itemize}
        replica $i$ generates block $B=\block{id, \hqc, \hqc.r+1,v_{cur}, txn}$ and invokes $\mvba_{v_{cur}}(B, \sigma)$ with predicate $f((B', \sigma'), \hqc)$ defined below, where $(B', \sigma')$ is the input of a replica.

        \item\label{async:mvba} {\bf Exit Fallback.}
        Upon $\mvba_v$ decides $B$ where $v\geq v_{cur}$, replica $i$ commits $B$ and all its ancestors, multicasts the threshold signature share $\thsig{ B.id, B.r, B.v }_i$, and updates $v_{cur}\gets v, r_{vote}\gets B.r, \mode\gets false$.
        Replica $i$ then waits until $B$ gets certified.

        
        \end{itemize}
    \vspace{2mm}
    
        Predicate $f((B', \sigma'), \hqc)$ of $\mvba_{v}$ (as replica $i$ with $\hqc$ in view $v$ when invoking $\mvba$) 
        \begin{itemize}[noitemsep,topsep=0pt]
            \item 
            returns $1$, if $B'.v=v$, $B'.qc.v=v-1$, $B'.r=B'.qc.r+1$ and  either (1) $\sigma'$ is a valid threshold signature on the message $(\texttt{proof}, B'.v, B'.qc)$ or (2) $\sigma'=\bot$ and $B'.qc.rank\geq \hqc.rank$; 
            \item 
            returns $0$ otherwise.
        \end{itemize}
        
    \end{mybox}

    \caption{\async}
    \vspace{-1em}
    \label{fig:ditto}
\end{figure*}

\paragraph{Description of Steady State.}
All the steps are described in an event-driven manner, i.e., a step is executed whenever a condition is satisfied during the protocol execution.
The protocol is similar to \sync with the following main differences.
The blocks do not contain TC anymore.
Each replica additionally keeps a boolean value $\mode$ to record if it is in the fallback. 
When the replica reaches timeout for a round in view $v$, it stops the timer and stops proposing or voting for any block of view $\leq v$.
The $1$-chain lock rule and $2$-chain commit rule still apply, but the two blocks in the $2$-chain need to have the same view number.

\paragraph{Description of Fallback.}
Now we give a brief description of the Fallback protocol, which replaces the Pacemaker protocol in the \sync protocol (Figure~\ref{fig:jolteon}).

Just like in \sync, when the timer expires, the replica tries to initiate the fallback (the equivalent of view-change) by broadcasting a timeout message containing the highest QC and a signature share of the current view number.
When receiving $2f+1$ timeout messages of the same view $v$, the replica enters the fallback, updates its current view number, and multicasts a \texttt{proof} message attaching the view number and its high-QC. On receiving a \texttt{proof} message with $qc$, the replica threshold-signs the message if $qc$ is no lower than its $\hqc$, otherwise sends back its $\hqc$ with a higher rank.
Then, the replica that sent \texttt{proof}, waits for either \texttt{proof} getting threshold-signed by $2f+1$ replicas, or receiving a higher-ranked QC.
In the former case, the replica inputs a block extending its $\hqc$ together with the threshold signature as the proof to the \mvba instance of the current view.
In the latter case, the replica inputs a block extending the received higher-ranked QC together with empty proof to the \mvba instance.
In both cases, the replica invokes \mvba with the state-aware predicate which also has its $\hqc$ as input.

When \mvba decides a block, the replica commits the block and its ancestors, multicasts a vote message to get the block certified, and exits the fallback. 
The replica waits until the block gets certified before executing other steps of the protocol.

\paragraph{Correctness of \async.}
The proof of safety and liveness for \async can be found in Appendix~\ref{sec:ditto-proof}.

\begin{theorem}[Efficiency]
    During the periods of synchrony with honest leaders, the amortized communication complexity per block decision is $O(n)$, and the \latency is $5$ rounds.
    During periods of asynchrony, the expected communication complexity per block decision is $O(n^2)$, and the expected \latency is $R+4$ rounds.
\end{theorem}

\begin{proof}
    When the network is synchronous and leaders are honest, no honest replica will multicast timeout messages. In every round, the designated leader multicast its proposal of size $O(1)$ (due to the use of threshold signatures for QC), and all honest replicas send the vote of size $O(1)$ to the next leader.
    Hence the communication cost is $O(n)$ per round and per block decision.
    For the block latency, since \sync adopts $2$-chain commit and need one more round for all replicas to receive the 2-chain proof, the latency is $2\times 2+1=5$ rounds.
    
    When the network is asynchronous and honest replicas enter the asynchronous fallback, each honest replica in the fallback only broadcast $O(1)$ number of messages, and each message has size $O(1)$. Hence, each instance of the asynchronous fallback has communication cost $O(n^2)$, and will commit a new block with probability $2/3$.
    Therefore, the expected communication complexity per block decision is $O(n^2)$.
    
    The latency to commit a block under fallback is $R+4$ rounds, consisting of $1$ round to exchange timeouts, $2$ rounds to multicast and ack \texttt{proof} message, $R$ rounds to finish \mvba, and $1$ round to get \mvba output certified.
    The latency can be further reduced to $R+3$ rounds by removing the round to certify the \mvba output, if the \mvba already certifies its output. 
\end{proof}

\section{Implementation and Evaluation}
\blue{
[Note to the reader: The implementation and evaluation of \vaba, and the asynchronous fallback path of \async are outdated.
The implementation and evaluation of \sync, and the synchronous path of \async are up to date. ]
}

\subsection{Implementation}
\label{sec:implementation}

We implement \sync and \async on top of a high-performance open-source implementation of HotStuff\footnote{
\url{https://github.com/asonnino/hotstuff/tree/3-chain}
}~\cite{yin2019hotstuff}. We selected this implementation because it implements a Pacemaker~\cite{yin2019hotstuff}, contrarily to the implementation used in the original HotStuff paper\footnote{\url{https://github.com/hot-stuff/libhotstuff}}.
Additionally, it provides well-documented benchmarking scripts to measure performance in various conditions, and it is close to a production system (it provides real networking, cryptography, and persistent storage). 
It is implemented in Rust, uses Tokio\footnote{\url{https://tokio.rs}} for asynchronous networking, ed25519-dalek\footnote{\url{https://github.com/dalek-cryptography/ed25519-dalek}} for elliptic curve based  signatures, and data-structures are persisted using RocksDB\footnote{\url{https://rocksdb.org}}. It uses TCP to achieve reliable point-to-point channels, necessary to correctly implement the distributed system abstractions.
We additionally use \texttt{threshold\_crypto}\footnote{\url{https://docs.rs/threshold_crypto/0.4.0/threshold_crypto/}} to implement random coins, Our implementations are between 5,000 and 7,000 LOC, and a further 2,000 LOC of unit tests. 
We are open sourcing our implementations of \sync
\footnote{
\ifdefined\cameraReady
\url{https://github.com/asonnino/hotstuff}
\else
Link omitted for blind review.
\fi
}
, and
\async and \vaba
\footnote{
\ifdefined\cameraReady
\url{https://github.com/danielxiangzl/hotstuff}
\else
Link omitted for blind review.
\fi
}.
We are also open sourcing all AWS orchestration scripts, benchmarking scripts, and measurements data to enable reproducible results
\footnote{
\ifdefined\cameraReady
\url{https://github.com/asonnino/hotstuff/tree/main/benchmark}, \url{https://github.com/danielxiangzl/hotstuff/tree/main/benchmark}
\else
Link omitted for blind review.
\fi
}.

\paragraph{\async with exponential backoff:}
From the protocol design of \vaba, we know that \vaba is exactly the asynchronous fallback of our \async with the timeout threshold \timeoutvalue set to be $0$. Therefore, under asynchrony or leader attacks, the performance of \vaba would be better than \async, as \vaba immediately proceeds to the next view without waiting for the timer to expire.
To improve the latency performance of \async under long periods of asynchrony or leader attacks, we adopt an exponential backoff mechanism for the asynchronous fallback as follows.
We say a replica executes the asynchronous fallback consecutively $x$ times if it only waits for the timer to expire for the first fallback, and skips waiting for the timer and immediately sends timeout for the rest $x-1$ fallbacks. 
Initially, replicas only execute asynchronous fallback consecutively $x=1$ time.
However, if a replica, within the timeout, does not receive from the steady state round-leader immediately after the fallback, it will multiply $x$ by a constant factor (5 in our experiments); otherwise, the replica resets $x=1$.
Therefore, during long periods of asynchrony or leader attacks, the number of consecutively executed fallbacks would be exponentially increasing ($1,5,25,...$); while during periods of synchrony and honest leaders, the number of consecutively executed fallbacks is always 1.

\subsection{Evaluation} 
\label{sec:evaluation}
We evaluate the throughput and latency of our implementations through experiments on Amazon Web Services (AWS). 
We particularly aim to demonstrate \first that \sync achieves the theoretically lower \latency than \threechain under no contention and \second that the theoretically larger message size during view-change does not impose a heavier burden, making \sync no slower than \threechain under faults (when the view-change happens frequently). Additionally we aim to show that \async adapts to the network condition, meaning that \third it behaves similarly to \sync when the network is synchronous (with and without faults) and \fourth close to our faster version of VABA (2-chain) when the adversary adaptively compromises the leader.

We deploy a testbed on Amazon Web Services, using
m5.8xlarge instances across 5 different AWS regions: N.
Virginia (us-east-1), N. California (us-west-1), Sydney (apsoutheast-2), Stockholm (eu-north-1), and Tokyo (ap-northeast1). They provide 10Gbps of bandwidth, 32 virtual CPUs (16
physical core) on a 2.5GHz, Intel Xeon Platinum 8175, and
128GB memory and run Linux Ubuntu server 20.04.

We measure throughput and \etelatency as the performance metrics.
Throughput is computed as the average number of committed transactions per second,
and \etelatency measures the average time to commit a transaction from the moment it is submitted by the client.
Compared with the \latency in our theoretical analysis, \etelatency also includes the queuing delay of the transaction when the clients' input rate is high which helps identify the capacity limit of our system.

\begin{figure}[t]
\centering
\includegraphics[width=0.48\textwidth]{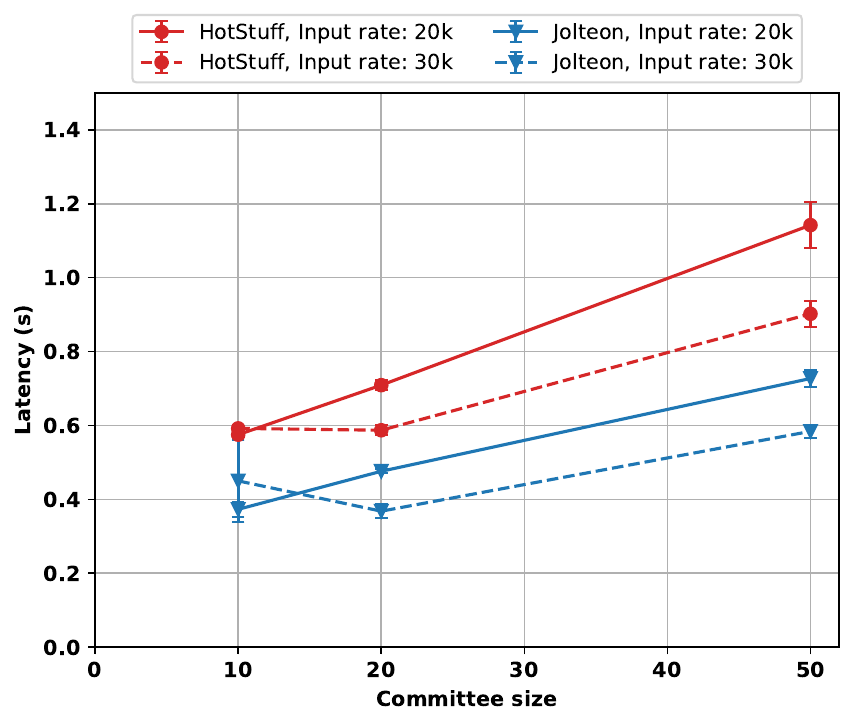}
\caption{ 
Comparative \latency for \threechain (HotStuff) and \sync. WAN measurements with 10, 20, or 50 replicas. No replica faults, 500KB mempool batch size and 512B transaction size.
}
\label{fig:happy-path-commit}
\end{figure}

\begin{figure*}[t]
\centering
\includegraphics[width=\textwidth]{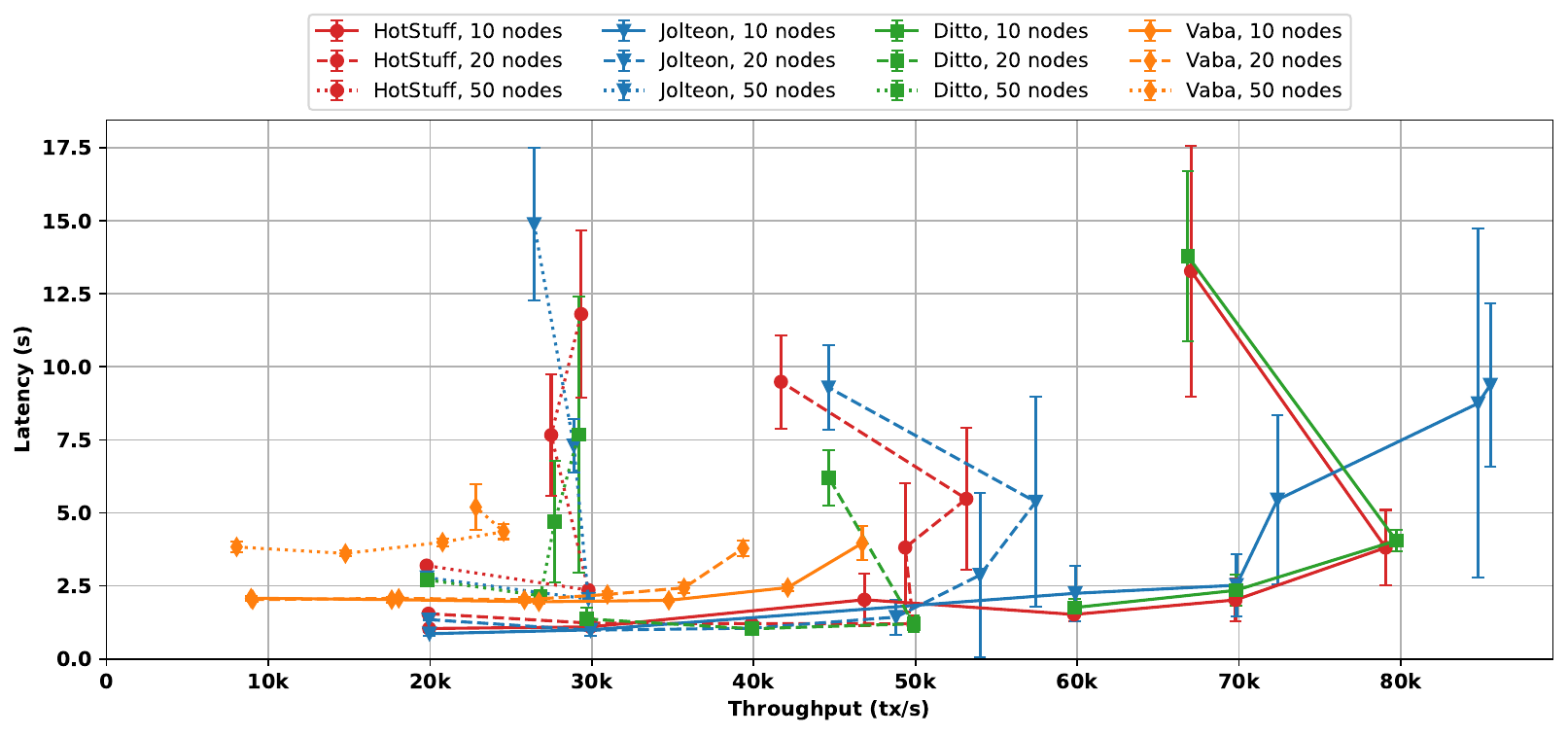}
\caption{ 
Comparative throughput-latency performance for \threechain (HotStuff), \sync, \async, and \vaba WAN measurements with 10, 20, or 50 replicas. No replica faults, 500KB mempool batch size and 512B transaction size.
}
\label{fig:happy-path}
\end{figure*}

\begin{figure*}[t]
\centering
\includegraphics[width=\textwidth]{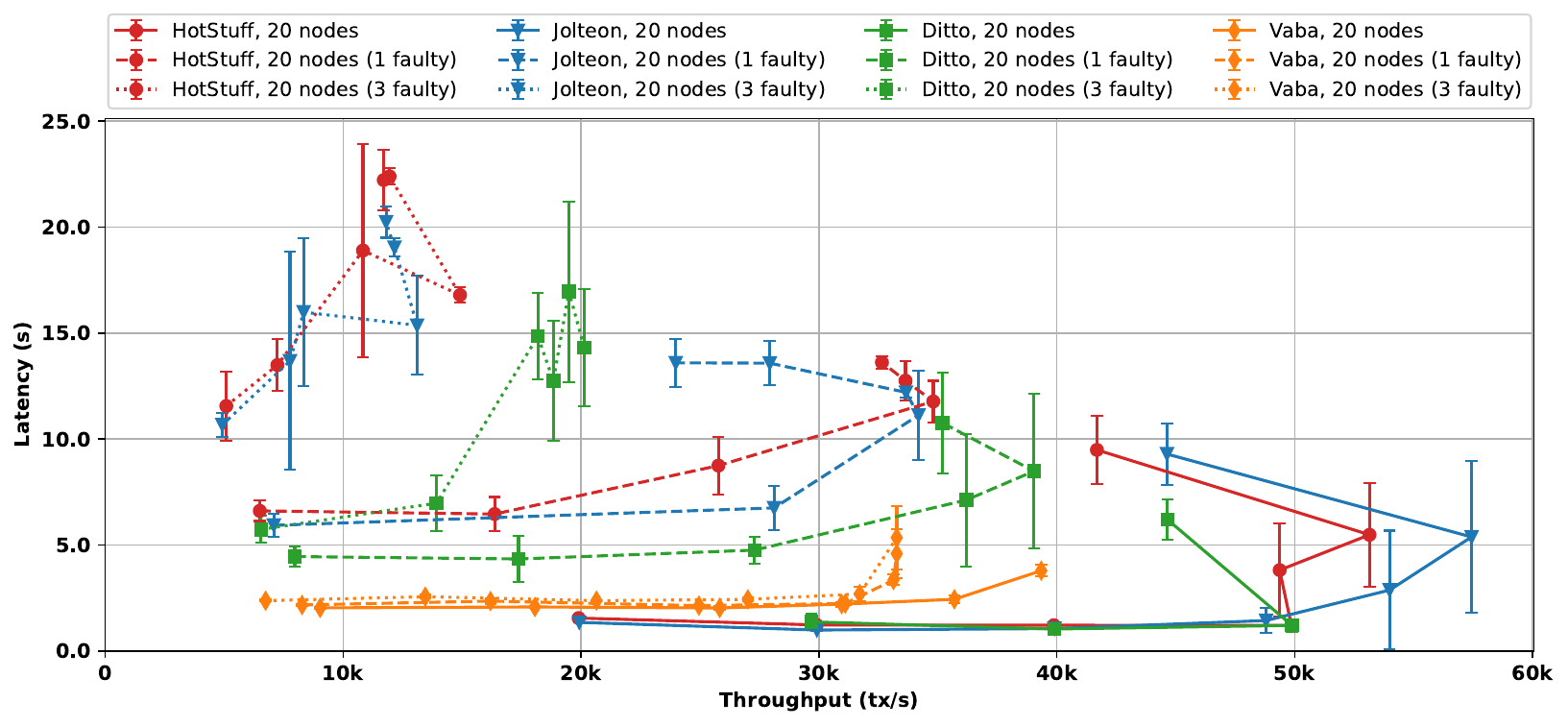}
\caption{ 
Comparative throughput-latency performance for \threechain (HotStuff), \sync, \async, and \vaba WAN measurements with 20 replicas. 0, 1, and 3 faults, 500KB mempool batch size and 512B transaction size.
}
\label{fig:dead-nodes}
\end{figure*}



\begin{figure}[h]
\centering
\includegraphics[width=0.48\textwidth]{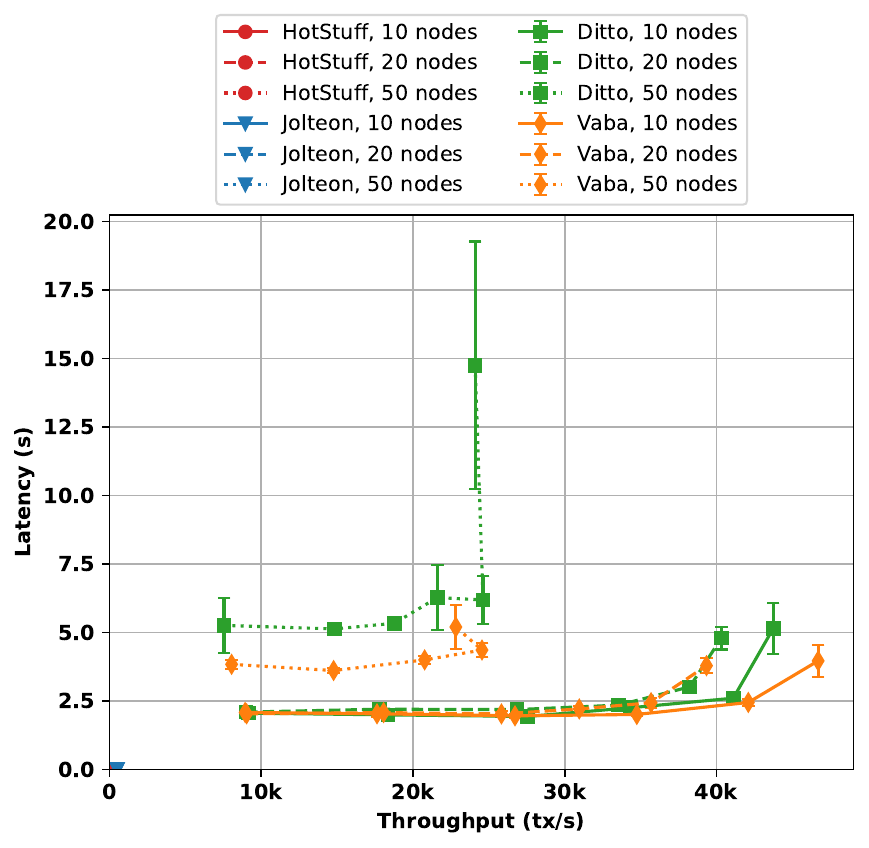}
\caption{ 
Comparative throughput-latency performance for \threechain (HotStuff), \sync, \async, and \vaba. WAN measurements with 10, 20, or 50 replicas. No replica faults, 500KB mempool batch size, and 512B transaction size. Leader constantly under DoS attack.
} 
\label{fig:leader-under-dos}
\end{figure}

\begin{figure}[h]
\centering
\includegraphics[width=0.48\textwidth]{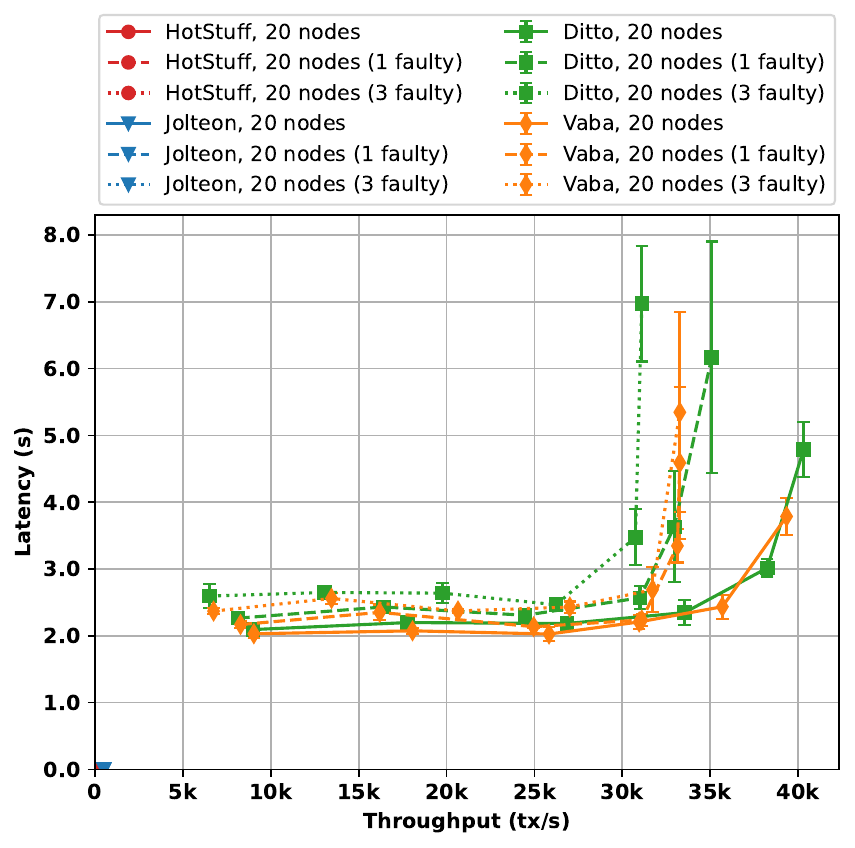}
\caption{ 
Comparative throughput-latency performance for \threechain (HotStuff), \sync, \async, and \vaba. WAN measurements with 20 replicas. 0, 1, and 3 faults, 500KB mempool batch size and 512B transaction size. Leader constantly under DoS attack.
}
\label{fig:dead-nodes-and-dos}
\end{figure}

In all our experiments, the transaction size is set to be $512$ bytes and the mempool batch size is set to be $500$KB. We deploy one benchmark client per node submitting transactions at a fixed rate for a duration of 5 minutes (to ensure we report steady state performance).
We set the timeout to be $5$ seconds for experiments with $10$ and $20$ nodes, and $10$ seconds for $50$ nodes, so that the timeout is large enough for not triggering the pacemaker of \sync and fallback of \async.
In the following sections, each measurement in the graphs
is the average of 3 runs, and the error bars represent one
standard deviation.

\subsection{Evaluation of \sync}
In this section, we compare \sync with our baseline \threechain implementation in two experiments.
First in Figure~\ref{fig:happy-path-commit} we run both protocol with a varying system size ($10$, $20$, $50$ nodes).  In order to remove any noise from the mempool, this graph does not show the \etelatency for clients but the time it takes for a block to be committed.  As the Figure illustrated \sync consistently outperforms \threechain by about $200-300$ms of latency which is around one round-trip across the world and both systems scale similarly. In Figure~\ref{fig:happy-path} this effect is less visible due to the noise of the mempool (\etelatency of around 2 secs), but \sync is still slightly faster than \threechain in most experiments.

Finally, in Figure~\ref{fig:dead-nodes} we run both protocols with 20 nodes and crashed 0, 1, or 3 nodes at the beginning of the experiment. This forces frequent view-changes due to the leader rotation that \threechain uses in order to provide fairness. This is ideal for our experiment since we can see the overall impact of the more costly view-change slowing down \sync. As we can see, \sync again outperforms \threechain in most settings due to the 2-chain commit enabling more frequent commits under faults. As a result, we can conclude that there is little reason to pay the extra round-trip of \threechain in order to have a theoretically linear view-change.
After sharing our findings with the Diem Team they are currently integrating an adaptation of \sync for their next mainnet release.

\subsection{Evaluation of \async}

\paragraph{Synchronous and fault-free executions.}
When all replicas are fault-free and the network is synchronous,
we compare the performances of the three protocol implementations in Figure~\ref{fig:happy-path}.
As we can observe from the figure, the synchronous path performance of \async is very close to that of \sync, when the quadratic asynchronous fallback of \async and the quadratic pacemaker of \sync is not triggered.
On the other hand, the performance of \vaba is worse than \sync and \async in this setting, due to its quadratic communication pattern -- instead of every replica receiving the block metadata and synchronizing the transaction payload with only one leader per round in \sync and \async, in VABA every replica will receive and synchronize with $O(n)$ leaders per round.

\paragraph{Crash faults.}
In this experiment, we run the three protocol implementations with $20$ nodes and crash $1$ or $3$ nodes from the beginning of the execution.
The throughput-latency results are summarized in Figure~\ref{fig:dead-nodes}.
As we can observe from the figure, 
\vaba has the most robust performance under faults, where the latency only slightly increases before the throughput exceeds $30k$ tps. 
The reason is that by design \vaba can make progress as long as at least $2f+1$ nodes are honest, and crashing some of the nodes only add latency to the views when these nodes are elected as the leader of the views.
The performance of \threechain and \sync are more fragile under faults, where the peak performance degrades to $30k$ tps with about $10$ seconds latency under $1$ fault and $10k$ tps with about $15$ seconds latency under $3$ faults. 
The reason is that whenever a round-$r$ leader is faulty, replicas need to wait for two timeouts, which is $10$sec since the timeout is set to be $5$sec, to enter round $r+1$ from round $r-1$. 
\async under faults performs better than \sync but worse than \vaba. Compared to \vaba, replicas in \async need to wait for a timeout of $5$sec to enter the asynchronous fallback whenever the leader is faulty; compared to \sync, replicas only need to wait for one timeout ($5$sec) to enter the fallback, which makes progress efficiently even under faults (as the \vaba under faults suggests).

\paragraph{Attacks on the leaders.}
Figure~\ref{fig:leader-under-dos} presents the measurement results.
When the eventual synchrony assumption does not hold, either due to DDoS attacks on the leaders or adversarial delays on the leaders' messages, \threechain and \sync will have no liveness, i.e., the throughput of the system is always $0$. The reason is that whenever a replica becomes the leader for some round, its proposal message is delayed and all other replicas will timeout for that round.
On the other hand, \async and \vaba are robust against such adversarial delays and can make progress under asynchrony.
The performance of the \vaba protocol implementation is not affected much by delaying a certain replica's proposal, as observed from the previous experiment with crashed replicas. Therefore, we use it as a baseline to compare with our \async protocol implementation.
Our results, confirm our theoretical assumption as the asynchronous fallback performance of \async is very close to that of \vaba under $10$ or $20$ nodes, and slightly worse than \vaba under $50$ nodes. This extra latency cost is due to the few timeouts that are triggered during the exponential back-off.

\paragraph{Crash faults and attacks on the leaders.}
Figure~\ref{fig:dead-nodes-and-dos} presents the measurement results.
The throughput of \threechain and \sync is again $0$ under leader attacks for the same reason mentioned above.
\vaba makes progress even under leader attacks, and its performance is also robust against 1 or 3 node crashes.
The performance of \async is very close to that of \vaba, since our implementation adopts the exponential backoff mechanism for the asynchronous fallback.
One interesting artifact of the exponential backoff is that, compared to the case under just crash faults and no DDoS (Figure~\ref{fig:dead-nodes}), \async has better performance under both crash faults and DDoS (Figure~\ref{fig:dead-nodes-and-dos}).
The reason is that under long periods of leader attacks, the \async will skip waiting for the time to expire for most of the views, and directly send timeouts and enter the fallback.

\paragraph{Take away.}
To conclude, there is little reason not to use \async as our experiments confirm our theoretical bounds. \async adapts to the network behavior and achieves almost optimal performance. The only system that sometimes outperforms \async is \vaba during intermittent periods of asynchrony as it does not pay the timeout cost of \async when deciding how to adapt. This, however, comes at a significant cost when the network is good and in our opinion legitimizes the superiority of \async when run over the Internet.

\section{Related Work}
\label{sec:relatedwork}

\paragraph{Eventually synchronous BFT.}
BFT SMR has been studied extensively in the literature.
A sequence of efforts~\cite{castro1999practical,buchman2016tendermint,buterin2017casper,gueta2019sbft,yin2019hotstuff,kogias2016enhancing} have been made to reduce the communication cost of the BFT SMR protocols, with the state-of-the-art being HotStuff~\cite{yin2019hotstuff} that has $O(n)$ cost for decisions, a 3-chain commit latency under synchrony and honest leaders, and $O(n^2)$ cost for view-synchronization. \sync presents another step forward from HotStuff as we realize the co-design of the pacemaker with the commit rules enables removing one round without sacrificing the linear happy path. Two concurrent theoretical works propose a 2-chain variation of the HotStuff as well~\cite{jalalzai2020fast, rambaud2020malicious}.
However, 
both protocols fail to realize the full power of 2-chain protocols missing the fact the view-change can become robust and DDoS resilient.

\paragraph{Asynchronous BFT.}
Several recent proposals focus on improving the communication complexity and latency, including HoneyBadgerBFT~\cite{miller2016honey}, VABA~\cite{abraham2019asymptotically}, Dumbo-BFT~\cite{guo2020dumbo}, Dumbo-MVBA~\cite{lu2020dumbo}, ACE~\cite{spiegelman2019ace}, Aleph~\cite{gkagol2019aleph}, and DAG-Rider~\cite{keidar2021all}.
A recent work 2PAC (2-phase asynchronous consensus)~\cite{cryptoeprint:2024:1108} achieves $9.5$ rounds as the expected latency and has $O(n^2)$ message complexity. 
The state-of-the-art protocols for asynchronous SMR have $O(n^2)$ cost per decision~\cite{spiegelman2019ace}, or amortized $O(n)$ cost per decision after transaction batching~\cite{guo2020dumbo, lu2020dumbo, gkagol2019aleph, keidar2021all}.

\paragraph{BFT with optimistic and fallback paths.}
To the best of our knowledge,~\cite{kursawe2005optimistic} is the first asynchronous BFT protocol with an efficient happy path.
Their asynchronous path has $O(n^3)$ communication cost while their happy path has $O(n^2)$ cost per decision, which was later extended~\cite{ramasamy2005parsimonious} to an amortized $O(n)$.
A recent paper~\cite{spiegelman2020search} further improved the communication complexity of asynchronous path to $O(n^2)$ and the cost of the happy path to $O(n)$.
The latency of these protocols is not optimized, 
  e.g. latency of the protocol in~\cite{spiegelman2020search} is $O(n)$.
Moreover, these papers are theoretical in nature and far from the realm of practicality.

Finally, a concurrent work named the Bolt-Dumbo Transformer (BDT)~\cite{lu2021bolt}, proposes a BFT SMR protocol with both synchronous and asynchronous paths and provides implementation and evaluation.
BDT takes the straightforward solution of composing three separate consensus protocols as black boxes. Every round starts with 1) a partially synchronous protocol (HotStuff), times-out the leader and runs 2) an Asynchronous Binary Agreement in order to move on and run 3) a fully asynchronous consensus protocol~\cite{guo2020dumbo} as a fallback.
Although BDT achieves asymptotically optimal communication cost for both paths this is simply inherited by the already known to be optimal back boxes. On the theoretical side, their design is beneficial since it provides a generally composable framework, but this generality comes at a hefty practical cost. BDT has a latency cost of 7 rounds (vs 5 of \async) at the fast path and of 45 rounds (vs 13.5 of \async) at the fallback, making it questionably practical. Finally, not opening the black-boxes stopped BDT from reducing the latency of HotStuff although it also has a quadratic view-change.


Our protocol, \async, has the asymptotically optimal communication cost in both the happy path and the asynchronous (fallback) path, but also good latency: in fact, due to the 2-chain design, the latency is even better than the state-of-the-art, (3-chain) HotStuff/DiemBFT and (3-chain) VABA protocols. \async also inherits the efficient pipelined design of these protocols (that are deployed in practice~\cite{baudet2019state}). In contrast with BDT, there is no overhead associated with switching between synchronous and asynchronous paths. We provide an implementation that organically combines the different modes of operation and extensive evaluation to support our claim that \async enjoys the best of both worlds with practical performance.


Another related line of work on optimistic BFT protocols for other network models, including
partially synchrony~\cite{abd2005fault, kotla2007zyzzyva, guerraoui2010next} and synchrony~\cite{pass2018thunderella, cryptoeprint:2018:980, nibesh2020optimality, cryptoeprint:2020:406}.
We leave the incorporation of an optimistic 1-chain commit path to \async to future work.

\paragraph{BFT protocols with flexible commit rules.}
Recently, a line of work studies and proposes BFT protocols with flexible commit rules for different clients of various beliefs on network and resilience assumptions, for the purpose of providing safety and liveness guarantees to the clients with the correct assumptions.
As a comparison, our work offers multiple commit paths for replicas under synchrony and asynchrony tolerating the same resilience threshold (one-thirds) and achieves optimal communication cost for each path asymptotically. 
Flexible BFT~\cite{malkhi2019flexible} proposes one BFT SMR solution supporting clients with various fault and synchronicity beliefs, but the guarantees only hold for the clients with the correct assumptions.
Strengthened BFT~\cite{xiang2021strengthened} shows how to strong commit blocks with higher resilience guarantees for partially synchronous BFT SMR protocols.
Several recent works~\cite{neu2020ebb, sankagiri2020blockchain} investigate how to checkpointing a synchronous longest chain protocol using partially synchronous BFT protocols, and offer clients with two different commit rules.

\section{Conclusion and Future Work}
We present \async, a practical byzantine SMR protocol that enjoys the best of both worlds: optimal communication on and off the happy path (linear and quadratic, respectively) and progress guarantees under the worst case asynchrony and DDoS attacks. As a secondary contribution, we design a 2-chain version of HotStuff, \sync, which leverages a quadratic view-change mechanism to reduce the latency of the standard 3-chain HotStuff.
We implement and experimentally evaluate all our systems to validate our theoretical analysis.

An interesting future work would be to incorporate an optimistic 1-chain commit path in our \sync and \async, which can further improve the latency performance of our systems when there are no faults.

\ifdefined\cameraReady
\section*{Acknowledgments}
This work is supported by the Novi team at Facebook. We
also thank the Novi Research and Engineering teams for
valuable feedback, and in particular Mathieu Baudet, Andrey
Chursin, George Danezis, Zekun Li, and Dahlia Malkhi for discussions
that shaped this work.

We thank Nibesh Shrestha, Isaac Doidge and Raghavendra Ramesh from Supra for pointing out an issue in the description of DiemBFT Pacemarker in the previous version of the paper. 
We thank Matthieu Rambaud for pointing out an issue in the DiemBFT timeout certificate definition in the previous version of the paper. 
The issues are fixed in this version.

We also thank Matthieu Rambaud for pointing out flaws of the \async protocol and its byproduct 2-chain VABA in the previous version of the paper. 
The \sync protocol remains secure. 
Please see~\cite{cryptoeprint:2024:1108} for the attacks on 2-chain VABA and the new protocol 2PAC that achieves $O(n^2)$ message complexity and $9.5\delta$ expected latency. 
For \async, we fixed the protocol in this version. 

\fi

\bibliographystyle{plain}
\bibliography{references}

\begin{thebibliography}{10}

\bibitem{abd2005fault}
Michael Abd-El-Malek, Gregory~R Ganger, Garth~R Goodson, Michael~K Reiter, and Jay~J Wylie.
\newblock Fault-scalable byzantine fault-tolerant services.
\newblock In {\em Proceedings of the twentieth ACM Symposium on Operating Systems Principles (SOSP)}, pages 59--74, 2005.

\bibitem{abraham2019communication}
Ittai Abraham, TH~Hubert Chan, Danny Dolev, Kartik Nayak, Rafael Pass, Ling Ren, and Elaine Shi.
\newblock Communication complexity of byzantine agreement, revisited.
\newblock In {\em Proceedings of the 2019 ACM Symposium on Principles of Distributed Computing}, pages 317--326, 2019.

\bibitem{abraham2019asymptotically}
Ittai Abraham, Dahlia Malkhi, and Alexander Spiegelman.
\newblock Asymptotically optimal validated asynchronous byzantine agreement.
\newblock In {\em Proceedings of the 2019 ACM Symposium on Principles of Distributed Computing (PODC)}, pages 337--346, 2019.

\bibitem{blum2020network}
Erica Blum, Jonathan Katz, and Julian Loss.
\newblock Network-agnostic state machine replication.
\newblock {\em arXiv preprint arXiv:2002.03437}, 2020.

\bibitem{buchman2016tendermint}
Ethan Buchman, Jae Kwon, and Zarko Milosevic.
\newblock The latest gossip on bft consensus.
\newblock {\em arXiv preprint arXiv:1807.04938}, 2018.

\bibitem{buterin2017casper}
Vitalik Buterin and Virgil Griffith.
\newblock Casper the friendly finality gadget.
\newblock {\em arXiv preprint arXiv:1710.09437}, 2017.

\bibitem{cachin2001secure}
Christian Cachin, Klaus Kursawe, Frank Petzold, and Victor Shoup.
\newblock Secure and efficient asynchronous broadcast protocols.
\newblock In {\em Annual International Cryptology Conference}, pages 524--541. Springer, 2001.

\bibitem{cachin2005random}
Christian Cachin, Klaus Kursawe, and Victor Shoup.
\newblock Random oracles in constantinople: Practical asynchronous byzantine agreement using cryptography.
\newblock {\em Journal of Cryptology}, 18(3):219--246, 2005.

\bibitem{castro1999practical}
Miguel Castro and Barbara Liskov.
\newblock Practical byzantine fault tolerance.
\newblock In {\em Proceedings of the third symposium on Operating Systems Design and Implementation (NSDI)}, pages 173--186. USENIX Association, 1999.

\bibitem{cryptoeprint:2018:980}
T-H.~Hubert Chan, Rafael Pass, and Elaine Shi.
\newblock Pili: An extremely simple synchronous blockchain.
\newblock Cryptology ePrint Archive, Report 2018/980, 2018.

\bibitem{dwork1988consensus}
Cynthia Dwork, Nancy Lynch, and Larry Stockmeyer.
\newblock Consensus in the presence of partial synchrony.
\newblock {\em Journal of the ACM (JACM)}, 35(2):288--323, 1988.

\bibitem{fischer1985impossibility}
Michael~J Fischer, Nancy~A Lynch, and Michael~S Paterson.
\newblock Impossibility of distributed consensus with one faulty process.
\newblock {\em Journal of the ACM (JACM)}, 32(2):374--382, 1985.

\bibitem{gkagol2019aleph}
Adam G{\k{a}}gol, Damian Le{\'s}niak, Damian Straszak, and Micha{\l} {\'S}wi{\k{e}}tek.
\newblock Aleph: Efficient atomic broadcast in asynchronous networks with byzantine nodes.
\newblock In {\em Proceedings of the 1st ACM Conference on Advances in Financial Technologies (AFT)}, pages 214--228, 2019.

\bibitem{guerraoui2010next}
Rachid Guerraoui, Nikola Kne{\v{z}}evi{\'c}, Vivien Qu{\'e}ma, and Marko Vukoli{\'c}.
\newblock The next 700 bft protocols.
\newblock In {\em Proceedings of the 5th European Conference on Computer Systems (EuroSys)}, pages 363--376, 2010.

\bibitem{gueta2019sbft}
Guy~Golan Gueta, Ittai Abraham, Shelly Grossman, Dahlia Malkhi, Benny Pinkas, Michael Reiter, Dragos-Adrian Seredinschi, Orr Tamir, and Alin Tomescu.
\newblock Sbft: a scalable and decentralized trust infrastructure.
\newblock In {\em 2019 49th Annual IEEE/IFIP international conference on dependable systems and networks (DSN)}, pages 568--580. IEEE, 2019.

\bibitem{guo2020dumbo}
Bingyong Guo, Zhenliang Lu, Qiang Tang, Jing Xu, and Zhenfeng Zhang.
\newblock Dumbo: Faster asynchronous bft protocols.
\newblock In {\em Proceedings of the 2020 ACM SIGSAC Conference on Computer and Communications Security}, pages 803--818, 2020.

\bibitem{jalalzai2020fast}
Mohammad~M Jalalzai, Jianyu Niu, Chen Feng, and Fangyu Gai.
\newblock Fast-hotstuff: A fast and resilient hotstuff protocol.
\newblock {\em arXiv preprint arXiv:2010.11454}, 2020.

\bibitem{keidar2021all}
Idit Keidar, Eleftherios Kokoris-Kogias, Oded Naor, and Alexander Spiegelman.
\newblock All you need is dag.
\newblock In {\em Proceedings of the 2021 ACM Symposium on Principles of Distributed Computing (PODC)}, 2021.

\bibitem{kogias2016enhancing}
Eleftherios~Kokoris Kogias, Philipp Jovanovic, Nicolas Gailly, Ismail Khoffi, Linus Gasser, and Bryan Ford.
\newblock Enhancing bitcoin security and performance with strong consistency via collective signing.
\newblock In {\em 25th Usenix Security Symposium (Usenix Security 16)}, pages 279--296, 2016.

\bibitem{kokoris2020asynchronous}
Eleftherios Kokoris~Kogias, Dahlia Malkhi, and Alexander Spiegelman.
\newblock Asynchronous distributed key generation for computationally-secure randomness, consensus, and threshold signatures.
\newblock In {\em Proceedings of the 2020 ACM SIGSAC Conference on Computer and Communications Security}, pages 1751--1767, 2020.

\bibitem{kotla2007zyzzyva}
Ramakrishna Kotla, Lorenzo Alvisi, Mike Dahlin, Allen Clement, and Edmund Wong.
\newblock Zyzzyva: speculative byzantine fault tolerance.
\newblock In {\em Proceedings of twenty-first ACM Symposium on Operating Systems Principles (SOSP)}, pages 45--58, 2007.

\bibitem{kursawe2005optimistic}
Klaus Kursawe and Victor Shoup.
\newblock Optimistic asynchronous atomic broadcast.
\newblock In {\em International Colloquium on Automata, Languages, and Programming (ICALP)}, pages 204--215. Springer, 2005.

\bibitem{lu2021bolt}
Yuan Lu, Zhenliang Lu, and Qiang Tang.
\newblock Bolt-dumbo transformer: Asynchronous consensus as fast as pipelined bft.
\newblock {\em arXiv preprint arXiv:2103.09425}, 2021.

\bibitem{lu2020dumbo}
Yuan Lu, Zhenliang Lu, Qiang Tang, and Guiling Wang.
\newblock Dumbo-mvba: Optimal multi-valued validated asynchronous byzantine agreement, revisited.
\newblock In {\em Proceedings of the 39th Symposium on Principles of Distributed Computing}, pages 129--138, 2020.

\bibitem{malkhi2019flexible}
Dahlia Malkhi, Kartik Nayak, and Ling Ren.
\newblock Flexible byzantine fault tolerance.
\newblock In {\em Proceedings of the 2019 ACM Conference on Computer and Communications Security (CCS)}, pages 1041--1053, 2019.

\bibitem{miller2016honey}
Andrew Miller, Yu~Xia, Kyle Croman, Elaine Shi, and Dawn Song.
\newblock The honey badger of bft protocols.
\newblock In {\em Proceedings of the 2016 ACM SIGSAC Conference on Computer and Communications Security (CCS)}, pages 31--42, 2016.

\bibitem{cryptoeprint:2020:406}
Atsuki Momose, Jason~Paul Cruz, and Yuichi Kaji.
\newblock Hybrid-bft: Optimistically responsive synchronous consensus with optimal latency or resilience.
\newblock Cryptology ePrint Archive, Report 2020/406, 2020.

\bibitem{neu2020ebb}
Joachim Neu, Ertem~Nusret Tas, and David Tse.
\newblock Ebb-and-flow protocols: A resolution of the availability-finality dilemma.
\newblock {\em arXiv preprint arXiv:2009.04987}, 2020.

\bibitem{pass2018thunderella}
Rafael Pass and Elaine Shi.
\newblock Thunderella: Blockchains with optimistic instant confirmation.
\newblock In {\em Annual International Conference on the Theory and Applications of Cryptographic Techniques}, pages 3--33. Springer, 2018.

\bibitem{ramasamy2005parsimonious}
HariGovind~V Ramasamy and Christian Cachin.
\newblock Parsimonious asynchronous byzantine-fault-tolerant atomic broadcast.
\newblock In {\em International Conference On Principles Of Distributed Systems}, pages 88--102. Springer, 2005.

\bibitem{rambaud2020malicious}
Matthieu Rambaud.
\newblock Malicious security comes for free in consensus with leaders.
\newblock {\em IACR Cryptol. ePrint Arch.}, 2020:1480, 2020.

\bibitem{cryptoeprint:2024:1108}
Matthieu Rambaud.
\newblock Faster asynchronous blockchain consensus and mvba.
\newblock Cryptology ePrint Archive, Report 2024/1108, 2024.

\bibitem{sankagiri2020blockchain}
Suryanarayana Sankagiri, Xuechao Wang, Sreeram Kannan, and Pramod Viswanath.
\newblock Blockchain cap theorem allows user-dependent adaptivity and finality.
\newblock {\em arXiv preprint arXiv:2010.13711}, 2020.

\bibitem{nibesh2020optimality}
Nibesh Shrestha, Ittai Abraham, Ling Ren, and Kartik Nayak.
\newblock On the optimality of optimistic responsiveness.
\newblock In {\em Proceedings of the 2020 ACM SIGSAC Conference on Computer and Communications Security (CCS)}, page 839–857, 2020.

\bibitem{spiegelman2020search}
Alexander Spiegelman.
\newblock In search for a linear byzantine agreement.
\newblock {\em arXiv preprint arXiv:2002.06993}, 2020.

\bibitem{spiegelman2019ace}
Alexander Spiegelman and Arik Rinberg.
\newblock Ace: Abstract consensus encapsulation for liveness boosting of state machine replication.
\newblock In {\em 23rd International Conference on Principles of Distributed Systems (OPODIS)}, 2020.

\bibitem{baudet2019state}
The~LibraBFT Team.
\newblock State machine replication in the libra blockchain, 2020.
\newblock \url{https://developers.libra.org/docs/state-machine-replication-paper}.

\bibitem{xiang2021strengthened}
Zhuolun Xiang, Dahlia Malkhi, Kartik Nayak, and Ling Ren.
\newblock Strengthened fault tolerance in byzantine fault tolerant replication.
\newblock In {\em The 41st IEEE International Conference on Distributed Computing Systems (ICDCS)}, 2021.

\bibitem{yin2019hotstuff}
Maofan Yin, Dahlia Malkhi, Michael~K Reiter, Guy~Golan Gueta, and Ittai Abraham.
\newblock Hotstuff: Bft consensus with linearity and responsiveness.
\newblock In {\em Proceedings of the 2019 ACM Symposium on Principles of Distributed Computing (PODC)}, pages 347--356, 2019.

\bibitem{yurek2023long}
Thomas Yurek, Zhuolun Xiang, Yu~Xia, and Andrew Miller.
\newblock Long live the honey badger: Robust asynchronous $\{$DPSS$\}$ and its applications.
\newblock In {\em 32nd USENIX Security Symposium (USENIX Security 23)}, pages 5413--5430, 2023.

\end{thebibliography}

\appendix

\section{Correctness of \sync}
\label{sec:jolteon-proof}
\subsection{Safety}
We begin by formalizing some notation.
\begin{itemize}
   \item We call a block byzantine (honest) if it was proposed by a byzantine (honest) replica.
   \item We say that a block $B$ is \emph{certified} if a quorum certificate $QC_B$ exists.
   \item $B_i \longleftarrow QC_i \longleftarrow B_{i+1}$  means that the block $B_i$ is certified by the quorum certificate $QC_i$ which is contained in the block $B_{i+1}$.
   \item $B_i \longleftarrow^* B_j$  means that the block $B_j$ \emph{extends} the block $B_i$.
   That is, there is exists a sequence $B_i \longleftarrow QC_i \longleftarrow B_{i+1} \longleftarrow QC_{i+1} \cdots \longleftarrow QC_{j-1} \longleftarrow B_{j}$ 
\end{itemize}
\begin{definition}[Global direct-commit]
\label{def:commit}
We say that a block $B$ is \emph{globally direct-committed} if $f+1$ honest replicas each successfully perform the \textbf{Vote} step on block $B'$ proposal in round $B.r+1$, such that $B'.qc$ certifies $B$.
These \textbf{Vote} calls invoke \textbf{Lock}, setting $qc_{high} \gets B'.qc$,
  and return $f+1$ matching votes (that could be used to form a $QC_{B'}$ with $f$ other matching votes). 
\end{definition}


\begin{lemma}
\label{lem:locommit}
If an honest replica successfully performs the \textbf{Commit} step on block $B$ then $B$ is globally direct-committed.
\end{lemma}
\begin{proof}
By the \textbf{Commit} condition, there exists a chain $B \longleftarrow QC_{B} \longleftarrow B' \longleftarrow QC_{B'}$ with $B'.r = B.r + 1$.
The existence of $QC_{B'}$ implies that $f+1$ honest replicas did \textbf{Vote} for $B'$.
\end{proof}

The next lemma follows from the voting rules and the definition of global direct commit.
\begin{lemma}
\label{lem:highqc}
If a block $B$ is globally direct-committed then any higher-round TC contains $qc_{high}$ 
  of round at least $B.r$.
\end{lemma}
\begin{proof}
By Definition~\ref{def:commit}, $f+1$ honest replicas execute the \textbf{lock} step in round $B.r + 1$
  and set $qc_{high}$ to $B'.qc$ that certifies $B$ (so $B'.qc.r = B.r$, $B'$ is the block proposed in round $B.r+1$).
None of these honest replicas may have previously timed out in round $B.r + 1$,
  and timing out stops voting in a round (but the replicas voted for $B'$).

Since $qc_{high}$ is never decreased, a timeout message prepared by any of the above $f+1$ honest replicas
  in rounds $> B.r$ contains a high qc of round at least $B.r$.
By quorum intersection, timeout messages used to prepare the TC in any round $> B.r$ contain a message from one of these honest replicas, completing the argument.
\end{proof}

Due to quorum intersection, we have
\begin{observation}
\label{obs:cert}
If a block is certified in a round, no other block can gather $f+1$ honest votes in the same round. Hence, at most one block is certified in each round.
\end{observation}

We can now prove the key lemma
\begin{lemma}
\label{lem:keysafety}
For every certified block $B'$ s.t.\ $B'.r \geq B.r$ such that $B$ is globally direct-committed, $B \longleftarrow^* B'$. 
\end{lemma}
\begin{proof}
By Observation~\ref{obs:cert}, 
$B \longleftarrow^* B'$ for every $B'$ s.t.\  $B'.r = B.r$.

We now prove the lemma by induction on the round numbers $r' > B.r$.
\vspace{0.4mm}

\textbf{Base case:} 
Let $r' = B.r + 1$.
$B$ is globally direct-committed, so by Definition~\ref{def:commit}, there are $f+1$ honest replicas that prepare votes in round $r' = B.r + 1$ on some block $B_{r+1}$ such that $B \longleftarrow QC_{B} \longleftarrow B_{r+1}$.
By Observation~\ref{obs:cert}, only $B_{r+1}$ can be certified in round $r'$.

\vspace{0.4mm}
\textbf{Step:} We assume the Lemma holds up to round $r' - 1 > B.r$ and prove that it also holds for $r'$.
If no block is certified at round $r'$, then the induction step holds vacuously.
Otherwise, let $B'$ be a block certified in round $r'$ and let $QC_{B'}$ be its certificate.
$B$ is globally direct-committed, so by Definition~\ref{def:commit}, there are $f+1$ honest replicas that have
  locked high qc in round $B.r+1$.
One of these replicas, $v$, must also have prepared a vote that is included in $QC_{B'}$ (as QC formation requires $2f+1$ votes and there are $3f+1$ total replicas).

Let $B'' \longleftarrow QC_{B''} \longleftarrow B'$ and denote $r'' = B''.r = QC_{B''}.r$. 
There are two cases to consider, $r'' \geq r$ and $r'' < r$. In the first case, 
by the induction assumption for round $r''$, $B \longleftarrow^* B''$ and we are done.

In the second case, $r'' < r < r'$ (the right inequality is by the induction step), i.e., the rounds for $B''$ and $B'$ are not consecutive. Hence, $B'$ must contain a TC for round $r' - 1$. 
By Lemma~\ref{lem:highqc}, this TC contains a $qc_{high}$ with round $\geq r$.

Consider a successful call by an honest replica to vote for $B'$. 
The only way to satisfy the predicate to vote is to satisfy (2),
which implies ${B''}.r \geq qc_{high}.r \geq B.r$, which is a contradiction to $r'' < r$.
\end{proof}

As a corollary of Lemma~\ref{lem:keysafety} and the fact that every globally direct-committed block is certified, we have
\begin{theorem}
\label{thm:agreement}
For every two globally direct-committed blocks $B,B'$, either $B \longleftarrow^* B'$ or $B' \longleftarrow^* B$.  
\end{theorem}
Let's call a successful invocation of the \textbf{Commit} step by a replica a local direct-commit.
For every locally committed block, there is a locally direct-committed block that extends it, and due to Lemma~\ref{lem:locommit}, also a globally direct-committed block that extends it. 
Each globally committed block defines a unique prefix to the genesis block, 
so Theorem~\ref{thm:agreement} applies to all committed blocks.
Hence all honest replicas commit the same block at each position in the blockchain.

Furthermore, since all honest replicas commit the transactions in one block following the same order, honest replicas do not commit different transactions at the same log position. 

\subsection{Liveness}
\begin{lemma} 
\label{lem:trigger}
When an honest replica in round $<r$ receives a proposal
  for round $r$ from another honest replica, it enters round $r$.
\end{lemma} 
\begin{proof}
Recall that a well-formed proposal sent by an honest replica contains either a TC or QC of round $r-1$.
When an honest replica receives such a proposal message,
  it will advance the round and enter round $r$.
\end{proof}
\begin{lemma}
\label{lem:rounds}
If the round timeouts and message delays between honest replicas are finite,
  then all honest replicas keep entering increasing rounds.
\end{lemma}
\begin{proof}
Suppose all honest replicas are in round $r$ or above, and let $v$ be an honest replica in round $r$.

We first prove that some honest replica enters round $r+1$. 
If all $2f+1$ honest replicas time out in round $r$,
  then $v$ will eventually receive $2f+1$ timeout messages, form a TC and enter round $r+1$.
Otherwise, at least one honest replica, $v'$ -- not having sent a timeout message for round $r$
  -- enters round $r+1$. 
For this, $v'$ must have observed $qc$ or $TC$ of round $r$ and updated its $qc_{high}$ or $TC_{last}$ accordingly.

Suppose $v'$ updated its $qc_{high}$.
Since $qc_{high}$ is never decreased and included in timeout messages,
  if $v'$ times out in any round $>r$, 
  then its timeout message will trigger $v$ to enter a round higher than $r$.
Otherwise, $v'$ must observe a QC in all rounds $> r$.
In this case, an honest leader sends a proposal in some round $>r$.
That proposal will eventually be delivered to $v$, triggering it to enter a higher round by Lemma~\ref{lem:trigger}.

Suppose $v'$ updated its $TC_{last}$.
Recall that $TC_{last}$ is never decreased, or only updated to $\bot$ when receiving a higher-round $qc$.
Suppose $TC_{last}$ is never updated to $\bot$. Since $TC_{last}$ is included in timeout messages, if $v'$ times out in any round $>r$, then its timeout message will trigger $v$ to enter a round higher than $r$.
Otherwise, $v'$ must observe a QC in all rounds $> r$, similarly $v$ will enter a higher round by Lemma~\ref{lem:trigger}.
Suppose $TC_{last}$ is updated to $\bot$ due to a higher-round $qc$. Then the previous argument applies. 

\end{proof}
In an eventually synchronous setting, the system becomes synchronous after the the global stabilisation time (GST).
We assume a known upper bound $\Delta$ on message transmission delays among honest replicas
  (practically, a back-off mechanism can be used to estimate $\Delta$) and
  let $4\Delta$ be the local timeout threshold for all honest replicas in all rounds.

Rounds are consecutive, advanced by quorum or timeout certificates, and honest replicas wait for proposals in each round.
We first show that honest replicas that receive a proposal without the round timer expiring accept the proposal, 
  allowing the quorum of honest replicas to drive the system progress.
\begin{lemma}
\label{lem:optresp}
Let $r$ be a round such that no QC has yet been formed for it and in which no honest replica has timed out.
When an honest replica $v$ receives a proposal of an honest leader of round $r$,
  $v$ will vote for the proposal.
\end{lemma}
\begin{proof}
The predicate in the \textbf{Vote} step for a proposal with block $B$ in round $r$ checks that
  (1) round numbers are monotonically increasing, and
  (2a) either the block extends the QC of the previous round ($B.qc.r + 1 = r$), or 
  (2b) the round of extended qc ($B.qc.r$) isn't less than the maximum high qc round in the TC of the previous round.

For (1), by assumption none of the $2f+1$ honest replicas have timed out, so no TC could have been formed for round $r$.
Also by assumption, no QC has been formed. Hence, no honest replica may have entered or voted in a round larger than $r$.
Round $r$ has an honest leader, so when a honest replica executes \textbf{Vote} step the round $r$ proposal,
  it does so for the first time and with the largest voting round.

For (2), we consider two cases.
If $B.tc = \bot$, then by well-formedness of honest leader's proposal, the $B.qc$
  it extends must have round number $r-1$, rounds are consecutive, and condition (a) holds.

If $B.tc$ is not empty, then it is a TC for round $r-1$, formed based on $2f+1$ timeout messages.
In this case, $B.qc.r \geq \max\{qc_{high}.r \mid qc_{high} \in B.tc\}$ predicate determines whether the replica votes
  for the proposal.
Since the leader is honest, $B.qc$ is the $qc_{high}$ of the leader when the proposal was generated.
The predicate holds as the honest leader updates the $qc_{high}$ to have round at least as large
  as the $qc_{high}$ of each timeout messages it receives (separately or within a forwarded TC).
\end{proof}
We now show a strong synchronization for rounds with honest leaders.
\begin{lemma}
  \label{lem:roundsync}
Let $r$ be a round after GST with an honest leader.
Within a time period of $2 \Delta$ from the first honest replica entering round $r$,
  all honest replicas receive the proposal from the honest leader.  
\end{lemma}
\begin{proof}
When the first honest replica enters round $r$, if it is not the leader, it must have formed a TC for round $r-1$.
Let $v$ be the honest leader of round $r$.
Since honest replicas forward TC to the leader of the next round, 
  $v$ will receive the TC and advance to round $r$ within $\Delta$ time of the first honest replica entering round $r$.

Upon entering round $r$, $v$ multicasts a proposal, which is delivered within $\Delta$ time to all honest replicas.
\end{proof}

Liveness follows from the following
\begin{theorem}
  Let $r$ be a round after GST.
  Every honest replica eventually locally commits some block $B$ with $B.r > r$.
\end{theorem}
\begin{proof}
Since the leaders are determined by round-robin and the number of byzantine replicas is bounded by $f$,
we can find round $r' > r$ such that rounds $r', r'+1, r'+2$ all have honest leaders.

By Lemma~\ref{lem:rounds} honest nodes enter increasing rounds indefinitely.
Due to Lemma~\ref{lem:roundsync},
  all honest replicas receive round $r'$ proposal with block $B$ from the leader within
  $2\Delta$ time of starting their round $r'$ timer
  (by Lemma~\ref{lem:trigger} triggering ones that haven't yet entered round $r'$ to do so).
By Lemma~\ref{lem:optresp}, honest replicas accept the proposal and vote for it.
Within time $\Delta$ their votes are delivered to the leader of round $r'+1$,
  who forms a QC extending $B$ and sends a proposal with a block $B_{r'+1}$.
This proposal will be received by honest replicas within another $\Delta$ time,
  by Lemma~\ref{lem:trigger} triggering them to enter round $r'+1$ before the local timer
  of $4 \Delta$ for round $r'$ expires.

By Lemma~\ref{lem:optresp}, every honest replica accepts the proposal, prepares a vote
  and sends it to round $r'+2$ leader, who is also honest.  
At this point, since $f+1$ honest replicas voted for $B$,
  by Definition~\ref{def:commit}, $B$ is globally direct-committed.

Continuing the argument, the honest leader of round $r'+2$ receives the votes to form the round $r'+1$ QC
  after at most $5\Delta$ time of the first honest validator entering round $r'$.
It then prepares and sends round $r'+2$ proposal that extends $QC_{B_{r'+1}}$.
After at most another $\Delta$ time all honest replicas receive this proposal,
  leading them to enter round $r'+2$ (before local timer for round $r'+1$ expires\footnote{Analogous to round $r'$. Moreover, as the round $r'+1$ leader enters the round first, no replica spends more than $3\Delta$ time in round $r'+1$.})
  and locally direct-commit $B$.
\end{proof}

Since we assume that each client transaction will be repeatedly proposed by honest replicas until it is committed (see Section~\ref{sec:prelim}), eventually each client transaction will be committed by all honest replicas.

\section{Correctness of \async}
\label{sec:ditto-proof}
\subsection{Safety}
We follow the same notation from Appendix~\ref{sec:jolteon-proof} with the following additional notations.
\begin{itemize}
    \item We say a block $B$ is sync-committed if it is directly committed due to a 2-chain $B \longleftarrow QC \longleftarrow B' \longleftarrow QC'$ in step {\bf Commit}.
    \item We say a block $B$ is async-committed if it is directly committed due to being the output of \mvba in step {\bf Exit Fallback}.
\end{itemize}
We redefine the notion of global direct-commit as follows.
\begin{definition}[Global direct-commit]
\label{def:g-commit}
We say that a block $B$ is \emph{globally direct-committed} if $f+1$ honest replicas each successfully perform the \textbf{Vote} step on block $B'$ proposal in round $B.r+1$ \blue{and view $B.v$}, such that $B'.qc$ certifies $B$.
These \textbf{Vote} calls invoke \textbf{Lock}, setting $qc_{high} \gets B'.qc$,
  and return $f+1$ matching votes (that could be used to form a $QC_{B'}$ with $f$ other matching votes). 
\end{definition}

From the definition, any block that is sync-committed is also globally direct-committed.

\begin{lemma}
\label{lem:0}
    For two certified blocks $B,B'$, if $B.rank=B'.rank$ then $B=B'$.
\end{lemma}
\begin{proof}
    Consider any view $v$, according to the protocol specification, the lowest ranked block in each view is the output of \mvba.  By the Agreement property of \mvba, only one block can be certified in step {\bf Exit Fallback} of view $v$.
    For later rounds of view $v$, due to quorum intersection, only one block can be certified for each round.
\end{proof}

\begin{lemma}
\label{lem:1}
    If a block $B$ is globally direct-committed, then any certified block $B'$ s.t. $B'.v=B.v$ and $B'.r \geq B.r$ must extend $B$, i.e., $B \longleftarrow^* B'$.
\end{lemma}
\begin{proof}
    Follows from the proof of Lemma~\ref{lem:keysafety}.
\end{proof}

\begin{lemma}
\label{lem:2}
    If a block $B$ is globally direct-committed, then any block $B'$ decided by $\mvba_{B.v+1}$ must extend $B$, i.e., $B \longleftarrow^* B'$.
\end{lemma}
\begin{proof}
    Let $v=B.v$.
    By Definition~\ref{def:g-commit}, $f+1$ honest replicas execute the \textbf{lock} step and set $\hqc$ such that $\hqc.rank\geq B.rank$.
    None of these honest replicas may have previously timed out in round $B.r + 1$, and timing out stops voting in a round (but the replicas voted for $B'$).
    Since $\hqc$ is never decreased, any of the above $f+1$ honest replicas must have $\hqc.rank\geq B.rank$ when sending the timeout message of view $v$.
    By quorum intersection, any $2f+1$ timeout messages of view $v$ contain a message from one of these honest replicas. 
    Therefore, any honest replica that sets $\mode=true$ (by receiving $2f+1$ timeout messages of view $v$) must update its $\hqc$ such that $\hqc.rank\geq B.rank$.

    Let $B'$ be the block decided by $\mvba_{v+1}$, so $B'.v=v+1$.
    For the sake of contradiction, suppose $B'$ does not extend $B$.
    Let $B''$ be the parent block of $B'$, we have $B''$ also does not extend $B$. By the External Validity property of \mvba, $B'$ must satisfy $f((B',\sigma'),\hqc)=1$ for at least one honest replica $h$ where $h$ has $\hqc$ when invoking $\mvba_{v+1}$. According to the definition of $f$, we have $B''.v=v$.
    There are two cases.
    \begin{itemize}
        \item $B''.rank\geq \hqc.rank$. Since $h$ has $\hqc.rank\geq B.rank$ when invoking $\mvba_{v+1}$, we have $B''.rank\geq B.rank$. Since $B''.v=B.v=v$, by Lemma~\ref{lem:1}, $B''$ must extend $B$, contradiction. 
        \item There exists a threshold signature $\sigma'$ on $(\texttt{proof}, B'.v, B'.qc)$. According to the protocol specification ({\bf Ack}), at least $2f+1$ distinct replicas have their $\hqc$ such that $B''.rank=B'.qc.rank\geq \hqc.rank$ when handling the $(\texttt{proof}, B'.v, B'.qc)$ message.
        Since honest replicas will not decrease their $\hqc$, and any honest replica has $\hqc.rank\geq B.rank$ when entering fallback, we conclude $B''.rank\geq \hqc.rank\geq B.rank$.  Since $B''.v=B.v=v$, by Lemma~\ref{lem:1}, $B''$ must extend $B$, again contradiction. 
    \end{itemize}
\end{proof}

\begin{lemma}
\label{lem:3}
    If a block $B$ is async-committed, then any certified block $B'$ s.t. $B'.v=B.v$ and $B'.r \geq B.r$ must extend $B$, i.e., $B \longleftarrow^* B'$.
\end{lemma}
\begin{proof}
    Consider the view $B.v$, according to the protocol specification, the lowest ranked block in view $B.v$ is the decision of $\mvba_{B.v}$, i.e., block $B$. At least $2f+1$ honest replicas set their $\hqc$ such that $\hqc.rank\geq B.rank$ when $\mvba_{B.v}$ decides. 
    
    Suppose for the sake of contradiction, there exists a certified block $B'$ s.t. $B'.v=B.v$ and $B'.r \geq B.r$, but $B'$ does not extend $B$.
    Let $B''$ denote the ancestor block of $B'$ such that $B''.v=B.v$ and has the lowest rank. Let $B'''$ denote the parent block of $B''$, then $B'''.v<B.v$. Since $B$ is the lowest ranked block in view $B.v$, and honest replicas vote with increasing round numbers in the same view, by quorum intersection at least one honest replica vote for $B$ and then vote for $B''$. When $h$ votes for $B$, it updates its $\hqc$ such that $\hqc.rank\geq B.rank$. Then, $h$ will not vote for $B''$ since $B'''.rank<B.rank\leq \hqc.rank$, contradiction. 
    
\end{proof}

\begin{lemma}
\label{lem:4}
    If a block $B$ is async-committed, then any block $B'$ decided by $\mvba_{B.v+1}$ must extend $B$, i.e., $B \longleftarrow^* B'$.
\end{lemma}
\begin{proof}
    At least $2f+1$ honest replicas set their $\hqc$ such that $\hqc.rank\geq B.rank$ when $\mvba_{B.v}$ decides. 
    By a similar proof of Lemma~\ref{lem:2}, any block $B'$ decided by $\mvba_{B.v+1}$ must extend $B$.
\end{proof}

\begin{lemma}
\label{lem:key1}
    For every certified block $B'$ s.t.\ $B'.rank \geq B.rank$ such that $B$ is async-committed, $B \longleftarrow^* B'$. 
\end{lemma}
\begin{proof}
    We prove by induction on the view numbers.
    For the base case of view $B.v$, by Lemma~\ref{lem:3}, any certified block of view $B.v$ with rank $\geq B.rank$ must extend $B$.

    For the induction step, suppose the lemma holds up to view $v\geq B.v$, and we will prove it holds also for view $v+1$.
    By the induction assumption, the async-committed block $B_v$ of view $v$ extends $B$, since $B''$ is certified.
    By Lemma~\ref{lem:4}, the async-committed block $B_{v+1}$ of view $v+1$ extends $B_v$, which implies that $B_{v+1}$ extends $B$.
    Then by Lemma~\ref{lem:3}, any certified block of view $v+1$ extends $B_{v+1}$, which also extends $B$.

    By induction, the lemma holds.
\end{proof}

\begin{lemma}
\label{lem:key2}
    For every certified block $B'$ s.t.\ $B'.rank \geq B.rank$ such that $B$ is globally direct-committed, $B \longleftarrow^* B'$. 
\end{lemma}
\begin{proof}
    By Lemma~\ref{lem:1}, any certified block of view $B.v$ and rank $\geq B.rank$ extends $B$.
    By Lemma~\ref{lem:2}, the async-committed block $B''$ of view $B.v+1$ extends $B$.
    Then, by Lemma~\ref{lem:key1}, any certified block $B'$ with $B'.v\geq B.v+1$ and $B'.rank \geq B.rank$ extends $B''$, which implies that $B'$ extends $B$.
\end{proof}

\begin{theorem}
\label{thm:agreement}
For every two committed blocks $B,B'$, either $B \longleftarrow^* B'$ or $B' \longleftarrow^* B$.  
\end{theorem}
\begin{proof}
    Let $B_1, B_2$ be any two blocks each directly committed due to either a sync-commit or an async-commit.
    By Lemma~\ref{lem:key1} and~\ref{lem:key2}, the fact that every sync-committed block is globally direct-committed, and the fact that every sync-committed or async-committed block is certified, we have $B_1,B_2$ extending one another.
    For any committed block, there exists a block extending it and gets either sync-committed or async-committed.
    Therefore, any two committed blocks also extend one another.
\end{proof}

Hence all honest replicas commit the same block at each position in the blockchain.
Furthermore, since all honest replicas commit the transactions in one block following the same order, honest replicas do not commit different transactions at the same log position. 

\subsection{Liveness}

\begin{lemma}
\label{lem:liveness:sync}
    If the network is synchronous and all replicas are honest, then all honest replicas keep committing new blocks with increasing round numbers.
\end{lemma}
\begin{proof}
    For the initial round (say round $0$), the leader proposes a genesis block $B_0$.
    Since the network is synchronous and all replicas are honest, all replicas will receive and vote for the block, and the votes will be received by the next leader. 
    Similarly, the next round-$1$ leader proposes the round-$1$ block $B_1$ extending $B_0$, gets voted by all honest replicas, and then the round-$2$ leader proposes the round-$2$ block $B_2$ extending $B_1$. By simple induction, all leaders keep proposing blocks and all blocks will be certified. Then, according to the commit rule, all replicas keep committing new blocks with increasing round numbers.
\end{proof}

\begin{lemma}
\label{lem:livenss:mvba}
    If an honest replica $h$ invokes $\mvba_v(B,\sigma)$, then for any honest replica $h'$, its predicate $f((B,\sigma), \hqc)=1$ where $\hqc$ is the high-QC of $h'$ when invoking $\mvba_v$. 
\end{lemma}
\begin{proof}
    Let $B'$ be the highest committed block with view $\leq v-1$.
    If $B'$ is sync-committed due to a 2-chain $B'\longleftarrow B''.qc \longleftarrow B''$, then from the proof of Lemma~\ref{lem:2}, any honest replica has its $\hqc.rank\geq B'.rank$ when entering the fallback (setting $\mode=true$). Also, $B''$ is the highest certified block in view $B'.v$, otherwise $B''$ should be the highest committed block. Note that $B''.r = B'.r+1$.
    If $B'$ is async-committed, then by the Agreement property of \mvba and step {\bf Exit Fallback}, all honest replicas have their $\hqc.rank\geq B'.rank$ when entering the fallback (setting $\mode=true$). Similarly, either $B'$ is the highest certified block in view $B'.v$, or there exists a higher certified block $B''$ in view $B'.v$ such that $B'\longleftarrow B''.qc \longleftarrow B''$. Note that $B''.r = B'.r+1$.
    Therefore, all honest replicas have their $\hqc.rank\geq B'.rank$ when entering the fallback (setting $\mode=true$), and any QC of view $B'.v$ has round $\leq B'.r+1$.

    If the honest replica $h$ invokes $\mvba_v(B,\sigma)$ with $\sigma\neq \bot$, it is externally valid for any honest replica according to the definition of predicate.
    
    If the honest replica $h$ invokes $\mvba_v(B,\sigma)$ with $\sigma= \bot$, according to the protocol specification ({\bf MVBA}), we have $B.qc.rank>\hqc.rank\geq B'.rank$ where $\hqc$ is the high-QC of $h$ when $h$ enters the fallback (sets $\mode=true$). As proved, any QC of view $B'.v$ has round $\leq B'.r+1$. Thus, $B.qc.rank\geq \hqc.rank$ where $\hqc$ is the high-QC of any honest replica when invoking $\mvba_v$. 
    According to the definition of predicate, $f((B,\bot), \hqc)=1$ since $B.qc.rank\geq \hqc.rank$ for any honest replica with $\hqc$ when invoking $\mvba_v$.

\end{proof}

\begin{lemma}
\label{lem:liveness:async}
    If the network is asynchronous, all honest replicas keep committing new blocks with increasing ranks.
\end{lemma}
\begin{proof}
    If at least one honest replica keeps committing new blocks in the {\bf Commit} step, then the committed blocks have increasing round numbers and all honest replicas eventually receive and commit these blocks.
    Otherwise, let $v$ be the highest view that any block is committed, and suppose block $B$ is async-committed in view $v$ ($\mvba_v$ decides this block). By the Agreement and Termination properties of \mvba, all honest replicas eventually decide $B$ and enter view $v$. 
    Since no honest replica commits new blocks, eventually all honest replicas will timeout view $v$, and enter fallback. Since the steps {\bf Enter Fallback, Ack, MVBA} are non-blocking, all honest replicas eventually invoke $\mvba_{v+1}$.
    By Lemma~\ref{lem:livenss:mvba}, any honest replica's input is externally valid to all replicas.
    By the Termination property of \mvba, all honest replicas will eventually decide one block from $\mvba_{v+1}$, which has a higher rank than previously committed blocks. Hence, all honest replicas keep committing new blocks with increasing ranks. 
\end{proof}

\begin{theorem}
    Each client transaction is eventually committed by all honest replicas.
\end{theorem}
\begin{proof}
    By Lemma~\ref{lem:liveness:sync} and~\ref{lem:liveness:async}, all honest replicas keep committing new blocks with increasing ranks.
    Since we assume that each client transaction will be repeatedly proposed by honest replicas until it is committed (see Section~\ref{sec:prelim}), eventually each client transaction will be committed by all honest replicas.
\end{proof}



\end{document}